\documentclass[a4paper,USenglish,cleveref, autoref, thm-restate]{lipics-v2021}

\usepackage{enumerate}
\usepackage{amsmath}
\usepackage{todonotes}
\usepackage{xspace}
\usepackage{pdfpages}

\newcommand{\defparprob}[4]{
  \vspace{1mm}
\noindent\fbox{
  \begin{minipage}{0.96\textwidth}
  \begin{tabular*}{\textwidth}{@{\extracolsep{\fill}}lr} #1  & {\bf{Parameter:}} #3 \\ \end{tabular*}
  {\bf{Input:}} #2  \\
  {\bf{Question:}} #4
  \end{minipage}
  }
  \vspace{1mm}
}

\newcommand{\defprob}[3]{
  \vspace{1mm}
\noindent\fbox{
  \begin{minipage}{0.96\textwidth}
  \begin{tabular*}{\textwidth}{@{\extracolsep{\fill}}lr} #1   \\ \end{tabular*}
  {\bf{Input:}} #2  \\
  {\bf{Question:}} #3
  \end{minipage}
  }
  \vspace{1mm}
}

\ifdefined\DEBUG{}
\newcommand{\mic}[1]{{\color{blue}{#1}}}
\def\rem#1{{\marginpar{\raggedright\scriptsize #1}}}
\newcommand{\micr}[1]{\rem{\textcolor{blue}{\(\bullet \) #1}}}
\newcommand{\bmp}[1]{{\color{red}{#1}}}
\newcommand{\bmpr}[1]{\rem{\textcolor{red}{\(\bullet \) #1}}}
\newcommand{\sk}[1]{{\color{magenta}{#1}}}
\newcommand{\skr}[1]{\rem{\textcolor{magenta}{\(\bullet \) #1}}}
\else
\newcommand{\mic}[1]{#1}
\newcommand{\micr}[1]{}
\newcommand{\bmp}[1]{#1}
\newcommand{\bmpr}[1]{}
\newcommand{\sk}[1]{#1}
\newcommand{\skr}[1]{}
\fi
\newcommand{\depr}[1]{}

\newcommand{\ssfull}{\textsc{Subset Sum}\xspace}
\newcommand{\ssshort}{\textsc{SS}\xspace}
\newcommand{\erbds}{\textsc{Exact Red-Blue Dominating Set}\xspace}
\newcommand{\erbdsshort}{\textsc{ERBDS}\xspace}
\newcommand{\rbds}{\textsc{Red-Blue Dominating Set}\xspace}
\newcommand{\rbdsshort}{\textsc{RBDS}\xspace}
\newcommand{\eewc}{\textsc{Exact-Edge-Weight Clique}\xspace}
\newcommand{\eewd}{\textsc{Exact-Edge-Weight $d$-Uniform Hyperclique}\xspace}
\newcommand{\eewdshort}{\textsc{EEW-$d$-HC}\xspace}
\newcommand{\eewcshort}{\textsc{EEWC}\xspace}
\newcommand{\ewcsp}{\textsc{Exact-Weight CSP}}
\newcommand{\mwcspg}{\textsc{Max-Weight CSP}$(\Gamma)$\xspace}
\newcommand{\ewcspg}{\textsc{Exact-Weight CSP}$(\Gamma)$\xspace}
\newcommand{\csp}{\textsc{CSP}}
\newcommand{\cspg}{\textsc{CSP}$(\Gamma)$\xspace}

\newcommand{\nphard}{\textsc{NP-}hard\xspace}

\newcommand{\ar}{\ensuremath{\text{\textsc{ar}}}\xspace}
\newcommand{\Oh}{\mathcal{O}}

\newcommand{\zz}{\mathbb{Z}}
\newcommand{\nn}{\mathbb{N}}
\newcommand{\eps}{\varepsilon}
\newcommand{\prob}[1]{\mathbb{P}\left(#1\right)}
\newcommand{\wsum}[1]{#1_{\textsc{sum}}}

\newcommand{\ncontainment}{\ensuremath{\mathsf{NP \not \subseteq coNP/poly}}\xspace}
\newcommand{\containment}{\ensuremath{\mathsf{NP \subseteq coNP/poly}}\xspace}
\newcounter{compositionCounter} 

\newtheorem{reductionrule}{Reduction Rule}

\title{On the Hardness of Compressing Weights}

\author{Bart M. P. Jansen}{Eindhoven University of Technology, The Netherlands \and \url{https://www.win.tue.nl/~bjansen/}}{b.m.p.jansen@tue.nl}{https://orcid.org/0000-0001-8204-1268}{} 

\author{Shivesh K. Roy}{Eindhoven University of Technology, The Netherlands \and \url{https://sites.google.com/view/shiveshroy}}{s.k.roy@tue.nl}{https://orcid.org/0000-0003-0896-3437}{}

\author{Micha\l \ W{\l}odarczyk}{Eindhoven University of Technology, The Netherlands \and \url{https://www.win.tue.nl/~mwlodarczyk/}} {m.wlodarczyk@tue.nl}{https://orcid.org/0000-0003-0968-8414}{}

\authorrunning{B.\,M.\,P. Jansen, S.\,K. Roy, and  M. W{\l}odarczyk}



\category{}

\relatedversion{}

\supplement{}

\funding{This project has received funding from the European Research Council (ERC) under the European Union's Horizon 2020 research and innovation programme (grant agreement No 803421, ReduceSearch).}

\nolinenumbers 

\hideLIPIcs  

\begin{document}

\maketitle

\begin{abstract}
We investigate computational problems involving large weights through the lens of kernelization, which is a~framework of polynomial-time preprocessing aimed at compressing the instance size.
Our main focus is the weighted \textsc{Clique} problem,
where we are given an edge-weighted graph and the goal is to detect a~clique of total weight equal to a~prescribed value.
We show that the weighted variant, parameterized by the number of vertices~$n$, is significantly harder than the unweighted problem 
by presenting an $\Oh(n^{3 - \varepsilon})$ lower bound on the size of the kernel, under the assumption that \ncontainment.
This lower bound is essentially tight: we show
that we can reduce the problem to the case with weights bounded by  $2^{\Oh(n)}$, which yields
a~randomized kernel of $\Oh(n^{3})$ bits.

We generalize these results to the weighted $d$-\textsc{Uniform Hyperclique} problem, \textsc{Subset Sum}, and weighted variants of Boolean \textsc{Constraint Satisfaction Problems} (CSPs). We also study weighted minimization problems and show that weight compression is easier when we only want to \bmp{preserve the collection of} optimal solutions. Namely, we show that for \bmp{node-}weighted \textsc{Vertex Cover} on bipartite graphs it is possible to maintain \bmp{the set of} optimal solutions using integer weights from the range $[1, n]$, but if we want to maintain the ordering of the weights of all inclusion-minimal solutions, then weights as large as $2^{\Omega(n)}$ are necessary.

\includegraphics[scale=0.1]{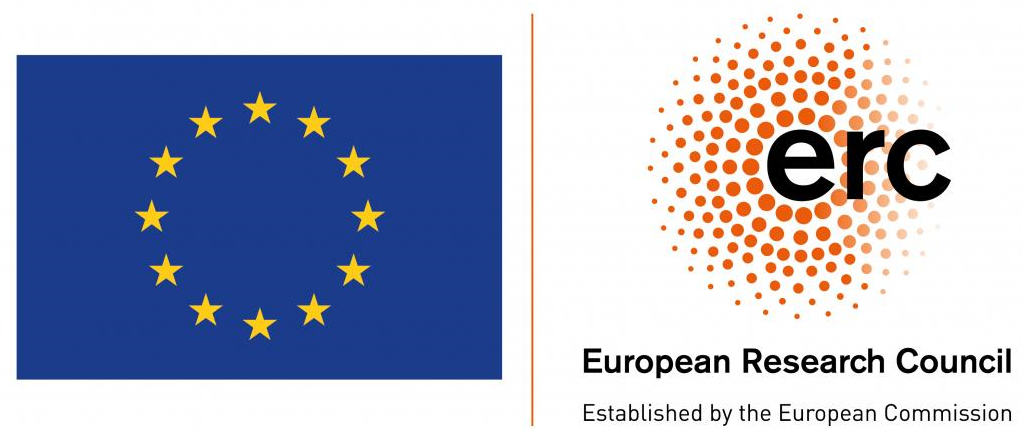}
\end{abstract}
\clearpage

\section{Introduction}
\label{sec:intro}

A prominent class of problems in algorithmic graph theory consist of finding a subgraph with certain properties in an input graph~$G$, if one exists. Some variations of this problem can be solved in polynomial time (detecting a triangle), while the general problem is NP-complete since it generalizes the \textsc{Clique} problem. In recent years, there has been an increasing interest in understanding the complexity of such subgraph detection problems in \emph{weighted} graphs, where either the vertices or the edges are assigned integral weight values, and the goal is either to find a subgraph of a given form which optimizes the total weight of its elements, or alternatively, to find a subgraph whose total weight matches a prescribed value. 

Incorporating weights in the problem definition can have a significant effect on computational complexity. For example, determining whether an unweighted $n$-vertex graph has a triangle can be done in time~$\Oh(n^{\omega})$ (where~$\omega < 2.373$ is the exponent of matrix multiplication)~\cite{ItaiR78}, while for the analogous weighted problem of finding a triangle of minimum edge-weight, no algorithm of running time~$\Oh(n^{3-\varepsilon})$ is known for any~$\varepsilon > 0$. Some popular conjectures in fine-grained complexity theory even postulate that no such algorithms exist~\cite{Williams15}. Weights also have an effect on the best-possible exponential running times of algorithms solving NP-hard problems: the current-fastest algorithm for the NP-complete \textsc{Hamiltonian Cycle} problem in undirected graphs runs in time~$\Oh(1.66^n)$~\cite{Bjorklund14}, while for its weighted analogue, \textsc{Traveling Salesperson}, no algorithm with running time~$\Oh((2-\varepsilon)^n)$ is known for general undirected graphs (cf.~\cite{Nederlof20}).

In this work we investigate how the presence of weights in a problem formulation affects the \emph{compressibility} and \emph{kernelization complexity} of NP-hard problems. Kernelization is a subfield of parameterized complexity~\cite{CyganFKLMPPS15,DowneyF13} that investigates how much a \emph{polynomial-time} preprocessing algorithm can compress an instance of an NP-hard problem, without changing its answer, in terms of a chosen complexity parameter. 

For a motivating example of kernelization, we consider the \textsc{Vertex Cover} problem. For the unweighted variant, a kernelization algorithm based on the Nemhauser-Trotter theorem~\cite{NemhauserT75} can efficiently reduce an instance~$(G,k)$ of the decision problem, asking whether~$G$ has a vertex cover of size at most~$k$, to an equivalent one~$(G',k')$ consisting of at most~$2k$ vertices, which can therefore be encoded in~$\Oh(k^2)$ bits via its adjacency matrix. In the language of parameterized complexity, the unweighted \textsc{Vertex Cover} problem parameterized by the solution size~$k$ admits a kernelization (self-reduction) to an equivalent instance on~$\Oh(k^2)$ bits. For the \emph{weighted} variant of the problem, where an input additionally specifies a weight threshold~$t \in \mathbb{N}_+$ and a weight function~$w \colon V(G) \to \mathbb{N}_{+}$ on the vertices, and the question is whether there is a vertex cover of size at most~$k$ and weight at most~$t$, the guarantee on the encoding size of the reduced instance is weaker. Etscheid et al.~\cite[Thm. 5]{EtscheidKMR17} applied a powerful theorem of Frank and Tard\"os~\cite{FrankT87} to develop a polynomial-time algorithm to reduce any instance~$(G,w,k,t)$ of \textsc{Weighted Vertex Cover} to an equivalent one with~$\Oh(k^2)$ edges, which nevertheless needs~$\Oh(k^8)$ bits to encode due to potentially large numbers occurring as vertex weights. The \textsc{Weighted Vertex Cover} problem, parameterized by solution size~$k$, therefore has a kernel of~$\Oh(k^8)$ bits. 

The overhead in the kernel size for the weighted problem is purely due to potentially large weights. This led Etscheid et al.~\cite{EtscheidKMR17} to ask in their conclusion whether this overhead in the kernelization sizes of weighted problems is necessary, or whether it can be avoided. As one of the main results of this paper, we will prove a lower bound showing that the kernelization complexity of some weighted problems is strictly larger than their unweighted counterparts.

\subparagraph*{Our results}

We consider an edge-weighted variation of the \textsc{Clique} problem, parameterized by the number of vertices~$n$:

\defprob{\eewc(\eewcshort)}{An undirected graph $G$, a weight function $w \colon E(G) \to \mathbb{N}_0$, and a target~$t \in \mathbb{N}_0$.}{Does $G$ have a clique of total edge-weight exactly $t$, i.e., a vertex set~$S \subseteq V(G)$ such that~$\{x,y\} \in E(G)$ for all distinct~$x,y \in S$ and such that~$\sum_{\{x, y\} \subseteq S}w(\{x,y\}) = t$?}

Our formulation of \eewcshort does not constrain the cardinality of the clique. This formulation will be convenient for our purposes, but we remark that by adjusting the weight function it is possible to enforce that any solution clique~$S$ has a prescribed cardinality. Through such a cardinality restriction we can obtain a simple reduction from the problem with potentially negative weights to equivalent instances with weights from~$\mathbb{N}_0$, by increasing all weights by a suitably large value and adjusting~$t$ according to the prescribed cardinality. \bmp{Note that an instance of \eewcshort can be reduced to an equivalent one where~$G$ has all possible edges, by simply inserting each non-edge with a weight of~$t+1$. Hence the difficulty of the problem stems from achieving the given target weight~$t$ as the total weight of the edges spanned by~$S$, not from the requirement that~$G[S]$ must be a clique.}

\eewcshort is a natural extension of \textsc{Zero-Weight Triangle}~\cite{AbboudFW20}, which has been studied because it inherits fine-grained hardness from both \textsc{3-Sum}~\cite{WilliamsW13} and \textsc{All Pairs Shortest Paths}~\cite[Footnote 3]{VassilevskaW09}. \eewcshort has previously been considered by Abboud et al.~\cite{AbboudLW14} as an intermediate problem in their W[1]-membership reduction from \textsc{$k$-Sum} to \textsc{$k$-Clique}. Vassilevska-Williams and Williams~\cite{WilliamsW13} considered a variation of this problem with weights drawn from a finite field. The related problem of detecting a triangle of negative edge weight is central in the field of fine-grained complexity for its subcubic equivalence~\cite{WilliamsW18} to \textsc{All Pairs Shortest Paths}. \bmp{Another example of an edge-weighted subgraph detection problem with an exact requirement on the weight of the target subgraph is \textsc{Exact-Edge-Weight Perfect Matching}, which can be solved using algebraic techniques~\cite[\S 6]{MulmuleyVV87} and has been used as a subroutine in subgraph isomorphism algorithms~\cite[Proposition 3.1]{MarxP13}.}

The unweighted version of \eewcshort, obtained by setting all edge weights to~$1$, is NP-complete because it is equivalent to the \textsc{Clique} problem. When using the number of vertices~$n$ as the complexity parameter, the problem admits a kernelization of size~$\Oh(n^2)$ obtained by simply encoding the instance via its adjacency matrix. We prove the following lower bound, showing that the kernelization complexity of the edge-weighted version is a~factor~$n$ larger. The lower bound even holds against \emph{generalized} kernelizations (see Definition~\ref{def:gen:kernel}).

\begin{theorem} \label{thm:eewc:lb}
The \eewc problem parameterized by the number of vertices~$n$ does not admit a generalized kernelization of~$\Oh(n^{3 - \varepsilon})$ bits for any~$\varepsilon > 0$, unless \containment.
\end{theorem}

Intuitively, the lower bound exploits the fact that the weight value of each of the~$\Theta(n^2)$ edges in the instance may be a large integer requiring $\Omega(n)$ bits to encode. We also provide a randomized kernelization which matches this lower bound.

\begin{theorem} \label{thm:eewc:ub}
There is a randomized polynomial-time algorithm that, given an $n$-vertex instance~$(G,w,t)$ of \eewc, outputs an instance~$(G',w',t')$ of bitsize~$\Oh(n^3)$, in which each number is bounded by~$2^{\Oh(n)}$, that is equivalent to~$(G,w,t)$ with probability at least~$1 - 2^{-n}$. Moreover, if the input is a YES-instance, then the output is always a YES-instance.
\end{theorem}

\mic{The proof is based on the idea that taking the weight function modulo a random prime
preserves the answer to the instance with high probability.
We adapt the argument by Harnik and Naor~\cite{HarnikN10} that it suffices to pick a prime of magnitude $2^{\Oh(n)}$.
As a result, each weight can be encoded with just $\Oh(n)$ bits.
}

It is noteworthy that the algorithm above can produce only false positives, therefore instead of using randomization we can turn it into a co-nondeterministic algorithm which guesses the correct values of the random bits.
The framework of cross-composition excludes not only deterministic kernelization, but also co-nondeterministic~\cite{DellM14}, thus the lower bound from Theorem~\ref{thm:eewc:lb} indeed makes the presented algorithm tight.

Together, Theorems~\ref{thm:eewc:lb} and~\ref{thm:eewc:ub} pin down the kernelization complexity of \eewc, and prove it to be a factor~$n$ larger than for the unit-weight case. For~\textsc{Clique}, the~kernelization of~$\Oh(n^2)$ bits due to adjacency-matrix encoding cannot be improved to~$\Oh(n^{2-\varepsilon})$ for any~$\varepsilon > 0$, as was shown by Dell and van Melkebeek~\cite{DellM14}.

We extend our results to the hypergraph setting, which is defined as follows: given a~$d$-regular hypergraph ($d \geq 3$) with \bmp{non-negative} integer weights on the hyperedges, and a~target value~$t$, test if there is a vertex set~$S$ for which each size-$d$ subset is a hyperedge (so that~$S$ is a hyperclique) such that the sum of the weights of the hyperedges contained in~$S$ is exactly~$t$. By a bootstrapping reduction using Theorem~\ref{thm:eewc:lb}, we prove that \eewd does not admit a generalized kernel of size~$\Oh(n^{d+1-\varepsilon})$ for any~$\varepsilon > 0$ unless \containment, while the randomized hashing technique yields a~randomized kernelization of size~$\Oh(n^{d+1})$.

We can view the edge-weighted ($d$-hyper)clique problem on~$(G,k,w,t)$ as a weighted constraint satisfaction problem (CSP) with weights from~$\mathbb{Z}$, by introducing a binary variable for each vertex, and a weighted constraint for each subset~$S'$ of~$d$ vertices, which is satisfied precisely when all variables for~$S'$ are set to \textsf{true}. If~$S'$ is a (hyper)edge~$e \in E(G)$ then the weight of the constraint on~$S'$ equals the weight of~$e$; if~$S'$ is not a hyperedge of~$G$, then the weight of the constraint on~$S'$ is set to~$- \infty$ to prevent all its vertices from being simultaneously chosen. Under this definition,~$G$ has a (hyper)clique of edge-weight~$t$ if and only if there is an assignment to the variables for which the total weight of satisfied constraints is~$t$. Via this interpretation, the lower bounds for \eewcshort yield lower bounds on the kernelization complexity of weighted variants of CSP. We employ a recently introduced framework~\cite{JansenW20} of reductions among different CSPs whose constraint languages have the same maximum degree~$d$ of their characteristic polynomials, to transfer these lower bounds to other CSPs \mic{(see Section~\ref{sec:csp} for definitions)}. We obtain tight kernel bounds when parameterizing the exact-satisfaction-weight version of CSP by the number of variables\bmp{, again using \mic{random prime numbers} to obtain upper bounds}.
\mic{Our lower bounds for \eewd transfer to all CSPs with degree~$d \geq 2$. 
In degree-1 CSP each constraint depends on exactly one variable, therefore its exact-weighted variant is equivalent to the \textsc{Subset Sum} problem, for which we also provide a tight lower bound.}

\begin{restatable}{theorem}{subsetSumState}
 \label{thm:ss:lowerbound}
\textsc{Subset Sum} parameterized by the number of items~$n$ does not admit a~generalized kernelization of size~$\Oh(n^{2-\varepsilon})$ for any~$\varepsilon > 0$, unless \containment.
\end{restatable}

Theorem~\ref{thm:ss:lowerbound} tightens a result of Etscheid et al.~\cite[Theorem 14]{EtscheidKMR17}, who ruled out (standard) kernelizations for \textsc{Subset Sum} of size~$\Oh(n^{2-\varepsilon})$ assuming the Exponential Time Hypothesis. Our reduction, conditioned on the incomparable assumption \ncontainment, additionally rules out generalized kernelizations that compress into an instance of a potentially different problem.
Note that the new lower bound implies that the input data in \textsc{Subset Sum} cannot be efficiently encoded in a more compact way, whereas the previous lower bound relies on the particular way the input is encoded in the natural formulation of the problem. 
\mic{On the other hand, a~randomized kernel of size $\Oh(n^2)$ is known~\cite{HarnikN10}.}

The results described so far characterize the kernelization complexity of broad classes of weighted constraint satisfaction problems in which the goal is to find a solution for which the total weight of satisfied constraints is exactly equal to a prescribed value. We also broaden our scope and investigate the maximization or minimization setting, in which the question is whether there is a solution whose cost is at least, or at most, a prescribed value. Some of our upper-bound techniques can be adapted to this setting: using a procedure by Nederlof, van Leeuwen and de Zwaan~\cite{NederlofLZ12} a maximization problem can be reduced to a polynomial number of exact queries. This leads, for example, to a \emph{Turing} kernelization (cf.~\cite{Fernau16}) for the weight-maximization version of \textsc{$d$-Uniform Hyperclique} which decides an instance in randomized polynomial time using queries of size~$\Oh(n^{d+1})$ to an oracle for an auxiliary problem. 
We do not have lower bounds in the maximization regime. 

\bmp{In an attempt to understand the relative difficulty of obtaining an exact target weight versus maximizing the target weight, we finally investigate different models of weight reduction for the \textsc{Weighted Vertex Cover} problem studied extensively in earlier works~\cite{ChlebikC08,EtscheidKMR17,NederlofLZ12}. We consider the problem on \emph{bipartite} graphs, where an optimal solution can be found in polynomial time, but \mic{we} investigate whether a weight function can be efficiently compressed while either preserving (a)~the collection of minimum-weight vertex covers, or (b)~the relative ordering of total weight for all \emph{inclusion-minimal} vertex covers. We give a polynomial-time algorithm for case (a) which reduces to a weight function with range~$\{1, \ldots, n\}$ using a relation to $b$-matchings, but show that in general it is impossible to achieve (b)~with a weight function with range~$\{1, \ldots, 2^{o(n)}\}$, by utilizing lower bounds on the number of \mic{different} threshold functions.}

\subparagraph*{Organization}

We begin with short preliminaries with the crucial definitions. 
We prove our main Theorem~\ref{thm:eewc:lb} in Section~\ref{sec:klb} by presenting a cross-composition of degree 3 into \eewc and employing it to obtain kernelization lower bounds for $d$-uniform hypergraphs for $d \ge 2$.
This section also contains
the kernelization lower bound for \ssfull as well as generalization of these results to Boolean CSPs.
Next, in Section~\ref{sec:vc:main} we focus on bipartite \textsc{Weighted Vertex Cover} and the difficulty of compressing weight functions.
The proofs of statements marked with $(\bigstar)$ are located in the appendix.
The kernel upper bounds,
including the proof of Theorem~\ref{thm:eewc:ub},
together with Turing kernelization for maximization problems, 
are collected in Appendix~\ref{sec:ub}.

\section{Preliminaries}
\label{sec:prelims}
We denote the set of natural numbers including zero by $\mathbb{N}_0$, and the set of positive natural numbers by $\mathbb{N}_+$. For positive integers~$n$ we define~$[n] := \{1, \ldots, n\}$. For a set~$U$ and integer~$d \geq 1$ we denote by~$\binom{U}{d}$ the collection of all size-$d$ subsets of~$U$. All logarithms we employ have base~$2$. 
Given a set $U$ and a weight function $w \colon U \to \mathbb{N}_0$, \bmp{for a subset $S \subseteq U$} we \bmp{denote} $w(S) := \sum _{v \in S} w(v)$.

All graphs we consider are undirected and simple. A (standard) graph~$G$ has a vertex set~$V(G)$ and edge set~$E(G) \subseteq \binom{V(G)}{2}$. For~$d \geq 2$, a $d$-uniform hypergraph~$G$ consists of a vertex set~$V(G)$ and a set of hyperedges~$E(G) \subseteq \binom{V(G)}{d}$, that is, each hyperedge is a set of exactly~$d$ vertices. Hence a $2$-uniform hypergraph is equivalent to a standard graph. A \emph{clique} in a $d$-uniform hypergraph~$G$ is a vertex set~$S \subseteq V(G)$ such that for each~$X \in \binom{S}{d}$ we have~$X \in E(G)$: each possible hyperedge among the vertices of~$S$ is present. A \emph{vertex cover} for a graph~$G$ is a vertex set~$S \subseteq V(G)$ containing at least one endpoint of each edge. A vertex cover is \emph{inclusion-minimal} if no proper subset is a vertex cover.

\subparagraph*{\bmp{Parameterized complexity}}
A \emph{parameterized problem}~$Q$ is a subset of $\Sigma^* \times \mathbb{N}_{+}$, where $\Sigma$ is a finite alphabet. 

\begin{definition} \label{def:gen:kernel}
Let $Q, Q' \subseteq \Sigma^*\times\mathbb{N}_{+}$ be parameterized problems and let $h \colon \mathbb{N}_{+}\rightarrow\mathbb{N}_{+}$  be a computable function. A \emph{generalized kernel for $Q$ into $Q'$ of size $h(k)$} is an algorithm that, on input $(x,k) \in \Sigma^*\times\mathbb{N}_{+}$, takes time polynomial in $|x|+k$ and outputs an instance $(x',k')$ such that:
\begin{enumerate}
\item $|x'|$ and $k'$ are bounded by $h(k)$, and
\item $(x',k')\in Q'$ if and only if $(x,k) \in Q$.
\end{enumerate}
The algorithm is a \emph{kernel} for $Q$ if $Q = Q'$. It is a \emph{polynomial (generalized) kernel} if $h(k)$ is a polynomial.
\end{definition}

\begin{definition}[\textbf{Linear-parameter transformations}]
Let $P$ and $Q$ be parameterized problems. We say that $P$ is \emph{linear-parameter transformable} to $Q$, if there exists a polynomial-time computable function $f \colon \Sigma^* \times \mathbb{N}_{+} \to \Sigma^* \times \mathbb{N}_{+}$, such that for all 
$(x, k) \in \Sigma^* \times \mathbb{N}_{+}$, (a) $(x, k) \in P$ if and only if $(x', k') = f(x, k) \in Q$ and (b) $k' \leq \mathcal{O}(k)$. The function $f$ is called a linear-parameter transformation.
\end{definition}

We employ a linear-parameter transformation for proving the lower bound for \ssfull.
For other lower bounds
we use the framework of cross-composition~\cite{BodlaenderJK14} directly.

\begin{definition}[\textbf{Polynomial equivalence relation}, {\cite[Def. 3.1]{BodlaenderJK14}}]
Given an alphabet $\Sigma$, an equivalence relation $\mathcal{R}$ on $\Sigma^\star$ is called a polynomial equivalence relation if the following conditions hold.
\begin{romanenumerate}
\item There is an algorithm that, given two strings $x, y \in \Sigma^ \star$, decides whether $x$
and $y$ belong to the same equivalence class in time polynomial in $|x| + |y|$.

\item For any finite set $S \subseteq \Sigma^ \star$ the equivalence relation $\mathcal{R}$ partitions the elements of $S$ into a number of classes that is polynomially bounded in the size of the largest element of $S$.
\end{romanenumerate}
\end{definition}

\begin{definition}[\textbf{Degree-$d$ cross-composition}] \label{def:deg-dcross-comp}
Let $L \subseteq \Sigma^ \star$ be a language, let $\mathcal{R}$ be a polynomial equivalence relation on $\Sigma^ \star$, and let $Q \subseteq \Sigma^ \star \times  \mathbb{N}_{+}$ be a parameterized problem. A~degree-d OR-cross-composition of $L$ into $Q$ with respect to $\mathcal{R}$ is an algorithm that, given $z$ instances $x_1, x_2, \ldots , x_z \in \Sigma^ \star$ of $L$ belonging to the same equivalence class of $\mathcal{R}$, takes time polynomial in $\sum_{i=1}^{z} |x_i|$ and outputs
an instance $(x', k') \in \Sigma^ \star \times \mathbb{N}_{+}$ such that:
\begin{romanenumerate}
\item the parameter $k'$ is bounded by $\mathcal{O}(z^{1/d} \cdot (\max_i|x_i|)^c)$, where $c$ is some constant independent of $z$, and
\item $(x', k') \in Q$ if and only if there is an $i \in [z]$ such that $x_i \in L$.
\end{romanenumerate}
\end{definition}

\begin{theorem} [{\cite[Theorem 3.8]{BodlaenderJK14}}] \label{thm:crosscomp:to:lowerbound}
Let $L \subseteq \Sigma ^\star$ be a language that is NP-hard under Karp reductions, let $Q \subseteq \Sigma^\star \times \mathbb{N}_{+}$ be a parameterized problem, and let $\varepsilon > 0$ be a real number. If~$L$ has a degree-$d$ OR-cross-composition into $Q$ and $Q$ parameterized by $k$ has a polynomial (generalized) kernelization of bitsize $\Oh(k^{d - \varepsilon})$, then \containment.
\end{theorem}

\section{Kernel lower bounds}
\label{sec:klb}

\subsection{\eewc}
In this section we show that \eewc parameterized by the number of vertices in the given graph $n$ does not admit a generalized kernel of size $\Oh(n^{3- \varepsilon})$, unless \containment.
We use the framework of cross-composition to establish a kernelization lower bound \cite{BodlaenderJK14}. We will use the \nphard \rbds (\rbdsshort) as a starting problem for the cross-composition. Observe that \rbdsshort is NP-hard because it is equivalent to \textsc{Set Cover} and \textsc{Hitting Set}~\cite{Karp72}. 

\defprob{\rbds(\rbdsshort)}{A bipartite graph $G$ with a bipartition of~$V(G)$ into sets~$R$ (red vertices) and~$B$ (blue vertices), and \bmp{a positive} integer~$d \leq |R|$.}{Does there exist a set $D \subseteq R$ with $|D| \leq d$ such that every vertex in $B$ has at least one neighbor in $D$?}

The following lemma forms the heart of the lower bound. It shows that an instance of \eewcshort on~$z \cdot N^{\Oh(1)}$ vertices can encode the logical OR of a sequence of~$z^3$ instances of size~$N$ each. Roughly speaking, this should be interpreted as follows: when~$z \gg N$, each of the roughly~$z^2$ edge weights of the constructed graph encodes~$z$ useful bits of information, in order to allow the instance on~$\approx z^2$ edges to represent all~$z^3$ inputs.

\begin{lemma} \label{lem:construction:rbds}
There is a polynomial-time algorithm that, given integers~$z,d,n,m$ and a set of~$z^3$ instances~$\{ (G_{i,j,k}, R_{i,j,k}, B_{i,j,k}, d) \mid i,j,k \in [z]) \}$ of RBDS such that~$|R_{i,j,k}| = m$ and~$|B_{i,j,k}| = n$ for each~$i,j,k \in [z]$, constructs an undirected graph~$G'$, integer~$t > 0$, and weight function~$w \colon E(G') \to \mathbb{N}_0$ such that:
\begin{enumerate}
    \item the graph~$G'$ contains a clique of total edge-weight exactly~$t$ if and only if there exist~$i^*,j^*,k^* \in [z]$ such that~$G_{i^*,j^*,k^*}$ has a red-blue dominating set of size at most~$d$,\label{construction:iff}
    \item the number of vertices in~$G'$ is~$\Oh(z (m + nd))$, and\label{construction:vertices}
    \item the values of~$t$ \bmp{and~$|V(G')|$} depend only on~$z,d,n$, and~$m$.\label{construction:target}
\end{enumerate}
\end{lemma}

\begin{proof}
We describe the construction of~$(G',w,t)$; it will be easy to see that it can be carried out in polynomial time. Label the vertices in each set~$R_{i,j,k}$ arbitrarily as~$r_1, \ldots, r_m$, and similarly label the vertices in each set~$B_{i,j,k}$ as~$b_1, \ldots, b_n$. We construct a graph~$G'$ with edge-weight function~$w$ and integer $t$ such that~$G'$ has a clique of total edge weight exactly $t$ if and only if some~$G_{i, j, k}$ is a YES-instance of~\rbdsshort. In the following construction we interpret edge weights as vectors of length $nz+1$ written in base~$(m+d+2)$, which will be converted to integers later. Starting from an empty graph, we construct~$G'$ as follows; see Figure~\ref{fig:clique}. 

\begin{enumerate}
    
    \item For each~$i \in [z]$, create a vertex~$b_i$. The vertices~$b_i$ form an independent set, so that any clique in~$G'$ contains at most one vertex~$b_i$. \label{step:b}
    
    \item For each~$j \in [z]$, create a vertex~set~$R_j = \{r_1^j, r_2^j, \cdots, r_m^j\}$ and insert edges of weight~$\Vec{0}$ between all possible~pairs of $R_j$.\label{step:r}
    
    \item For each~$k \in [z]$, create a vertex $s_k$.
    The vertices~$s_k$ form an independent set, so that any clique in~$G'$ contains at most one vertex~$s_k$. \label{step:k}
    
    \item For each~$j,k \in [z]$, for each~$x \in [m]$, insert an edge between~$s_k$ and~$r^j_x$ of weight~$\Vec{0}$. \label{step:edge:s:r}
    
    \setcounter{compositionCounter}{\value{enumi}}
\end{enumerate}

    The next step is to ensure that the neighborhood of a vertex~$r_x$ in $G_{i, j, k}$ is captured in the weights of the edges which are incident~on~$r_x^j$ in $G'$.

\begin{enumerate}
     \setcounter{enumi}{\value{compositionCounter}}
    
    \item For each~$i,j \in [z]$, for each~$x \in [m]$, insert an edge between~$b_i$ and~$r^j_x$.\label{step:edge:b:r} 
    
    \item The weight of each edge~$\{b_i, r_x^j\}$ is a vector of length $nz+1$, out of which the least significant $nz$ positions are divided into $z$ blocks of length $n$ each, and the most significant position is 1. The numbering of blocks as well as positions within a given block start with the least significant position. 
     
    For each~$i,j \in [z]$, for each~$x \in [m]$, the weight of edge~$\{b_i, r^j_x\}$ is defined as follows. For each~$k \in [z]$, for each~$q \in [n]$, the value~$v_{k,q}(b_i, r^j_x)$ represents the value of the $q^{th}$ position of the $k^{th}$ block of the weight of $\{b_i, r^j_x\}$. The value is defined based on the neighborhood of vertex $r_x$ in $G_{i, j, k}$ as follows:\label{step:edgewt:b:r}
    \begin{equation}
    v_{k,q}(b_i, r_x^j) = \begin{cases}
    1 & \mbox{if $\{b_q, r_x\} \in E(G_{i, j, k})$} \\
    0 & \mbox{otherwise.}
    \end{cases}
    \end{equation}
    Intuitively, the vector representing the weight of edge~$\{b_i, r^j_x\}$ is formed by a 1 followed by the concatenation of $z$ blocks of length~$n$, such that the $k^{th}$ block is the $0/1$-incidence vector describing which of the~$n$ blue vertices of instance~$G_{i,j,k}$ are adjacent to~$r_x$. 
\setcounter{compositionCounter}{\value{enumi}}
\end{enumerate}
Note that the~$n$ blue vertices of an input instance~$G_{i,j,k}$ are represented by a single blue vertex~$b_i$ in~$G'$. The difference between distinct blue vertices is encoded via different positions of the weight vectors. The most significant position of the weight vectors, which is always set to~$1$ for edges of the form~$\{b_i, r^j_x\}$, will be used to keep track of the number of red vertices in a solution to \rbdsshort. 

\begin{figure}[t]
    \centering
    \includegraphics[scale=0.8]{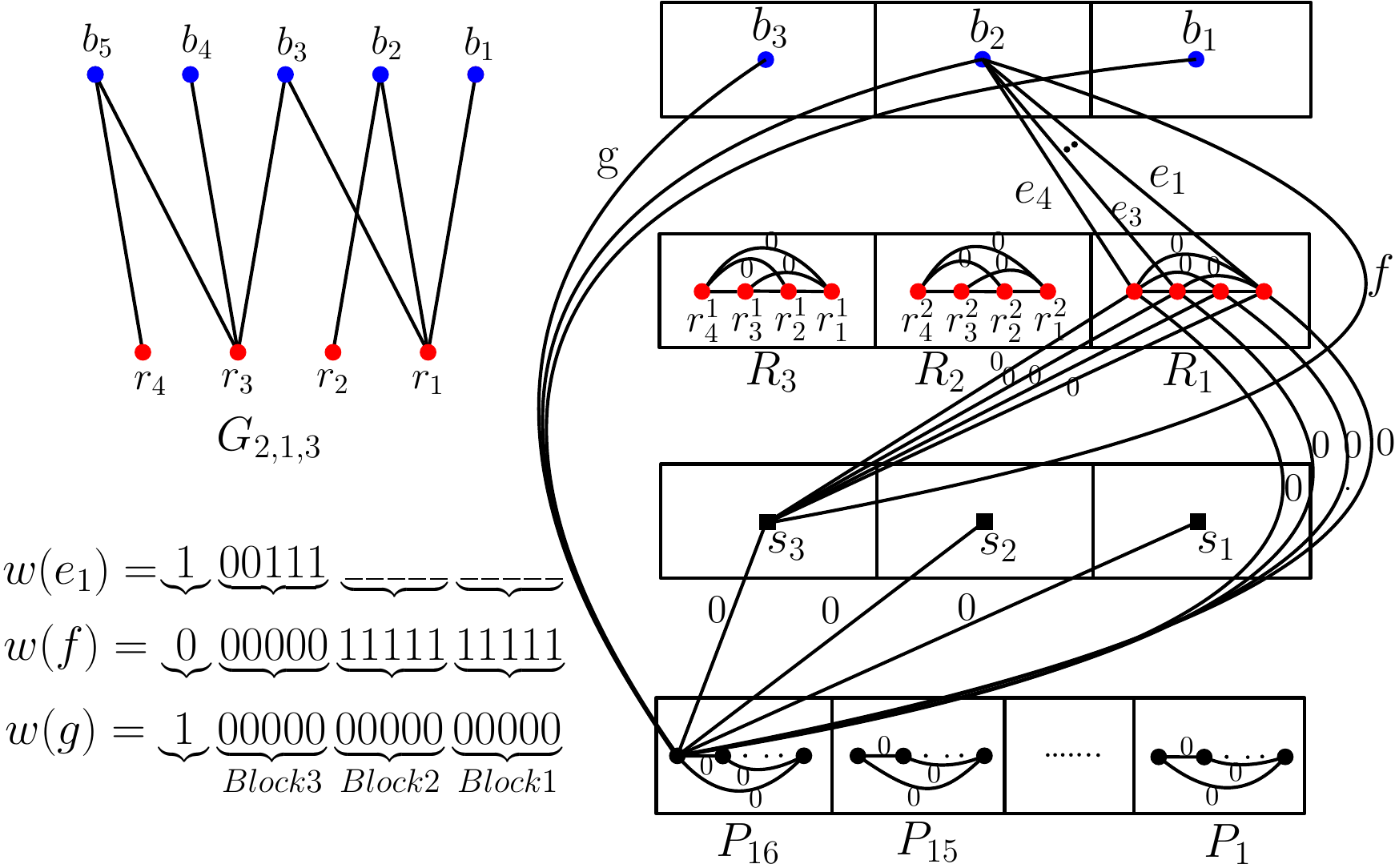}
    \caption{Top-left: An instance~$(G_{2, 1, 3}, R_{2, 1, 3}, B_{2, 1 , 3}, 2)$ of~\rbdsshort with~$m=4, n=5$, and~$d=4$. Right: Illustration of the \eewcshort instance created for a sequence of~$3^3$ inputs including the one on the left. For readability, only a subset of the edges is drawn. Bottom-left: For each type edge with non-zero weight, an example weight is shown in vector form.} \label{fig:clique}
\end{figure}

The graph constructed so far has a mechanism to select the first index~$i$ of an instance~$G_{i,j,k}$ (by choosing a vertex~$b_i$), to select the second index~$j$ (by choosing vertices~$r^j_x$), and to select the third index~$k$ (by choosing a vertex~$s_k$). The next step in the construction adds weighted edges~$\{b_i, s_k\}$, of which a solution clique in~$G'$ will contain exactly one. The weight vector for this edge is chosen so that the domination requirements from all RBDS instances whose third index differs from~$k$ (and which are therefore not selected) can be satisfied ``for free''. 

  
\begin{enumerate}
    \setcounter{enumi}{\value{compositionCounter}}
    \item For each~$i,k \in [z]$, insert an edge between~$b_i$ and~$s_{k}$. 
    
    \item As in Step \ref{step:edgewt:b:r}, the weight of the edge~$\{b_i, s_{k}\}$ is a ($1 + nz$)-tuple consisting of the most significant position followed by~$z$ blocks of length~$n$. There is a~$0$ at the most significant position, block $k$ consists of~$n$ zeros, and the other blocks are filled with ones. Hence the weight of the edge~$\{b_i, s_{k}\}$ is independent of~$i$.  \label{step:edge:bs}
    
    %
    
    
    \setcounter{compositionCounter}{\value{enumi}}
\end{enumerate}

    To be able to ensure that~$G'$ has a clique of exactly weight~$t$ if some input instance~$G_{i,j,k}$ has a solution, we need to introduce padding numbers which may be used as part of the solution to~\eewcshort. 
    
\begin{enumerate}
    \setcounter{enumi}{\value{compositionCounter}}
    \item For each position~$v \in [nz+1]$ of a weight vector, add a vertex set~$P_v= \{p_1^v, p_2^v, \cdots, p_{d-1}^v\}$ to~$G'$. Recall that~$d$ is the upper bound on the solution size for \rbdsshort. \label{step:padding}

    \item For each~$i \in [z]$, for each $v \in [nz+1]$, for each $y \in [d-1]$, add an edge~$\{b_i, p_y^v\}$. The weight of edge~$\{b_i, p_y^v\}$ has value~1 at the~$v^{th}$ position and \bmp{zeros elsewhere.} \label{step:edge:bp}
    
    \item For each~$v \in [nz+1]$, for each~$y \in [d-1]$, add an edge~$\{p_y^v, u\}$ of weight~$\vec{0}$ for all $u \in V(G') \setminus (\{b_i \mid i \in [z]\} \cup \{p_y^v\})$, i.e., for all vertices~$u \neq p_y^v$ which were not already adjacent to~$p_y^v$. \label{step:edge:b,rest}
\end{enumerate}

We define the target weight~$t$ to be the~$(nz+1)$-length vector with value~$d$ at each position, which satisfies Condition~\ref{construction:target}. Observe that $G'$ has $\Oh(z(m+nd))$ vertices: Steps~\ref{step:b} and~\ref{step:k} contribute~$\Oh(z)$ vertices, Step~\ref{step:r} contributes~$\Oh(z m)$, and Step~\ref{step:padding} contributes~$\Oh(d (nz))$. Hence Condition~\ref{construction:vertices} is satisfied. It remains to verify that $G'$ has a~clique of total edge weight exactly~$t$ if and only if some input instance~$G_{i, j, k}$ has a solution of \rbds of size at most~$d$. Before proving this property, we show the following claim which implies that no carries occur when summing up the weights of the edges of a clique in~$G'$.

\begin{claim}
For any clique~$S \subseteq V(G')$, for any position~$v \in [nz + 1]$ of a weight vector, there are at most~$d+m+1$ edges of the clique~$G'[S]$ whose weight vector has a~$1$ at position~$v$, and all other weight vectors are~$0$ at position~$v$.
\end{claim}
\begin{claimproof}
By construction, the entries of the vector encoding an edge weight are either~$0$ or~$1$.

By Steps~\ref{step:b} and~\ref{step:k}, a clique~$S$ in~$G'$ contains at most one vertex~$b_i$ and one vertex~$s_k$. Since~$G'$ does not have edges between vertices in distinct sets~$R_{j}$ and~$R_{j'}$ by Step~\ref{step:r}, any clique in~$G'$ consists of at most one vertex~$b_i$, one vertex~$s_k$, a subset of one set~$R_j$, and a subset of~$\bigcup _{v \in [nz + 1]} P_v$. For any fixed position~$v \in [nz + 1]$, the only edge-weight vectors which can have a~$1$ at position~$v$ are the~$d-1$ edges from~$P_v$ to~$b_i$, the edge~$\{b_i, s_k\}$, and the~$m$ edges between~$R_j$ and~$b_i$. As this yields~$(d-1) + 1 + m$ edges that possibly have a~$1$ at position~$v$, the claim follows.
\end{claimproof}

The preceding claim shows that when we convert each edge-weight vector to an integer by interpreting the vector as its base-$(m+d+2)$-representation, then no carries occur when computing the sum of the edge-weights of a clique. Hence the integer edge-weights of a clique~$S \subseteq V(G')$ sum to the integer represented by vector~$t$, if and only if the edge-weight vectors of the edges in~$S$ sum to the vector~$t$. In the remainder, it therefore suffices to prove that there is a YES-instance~$G_{i^*, j^*, k^*}$ of \rbdsshort among the inputs if and only if~$G'$ has a clique whose edge-weight vectors sum to the vector~$t$. We prove these two implications.

\begin{claim}
If some input graph~$G_{i^*, j^*, k^*}$ has a red-blue dominating set of size at most~$d$, then~$G'$ has a clique of edge-weight exactly~$t$.
\end{claim}
\begin{claimproof}
Let $S \subseteq R_{i^*, j^*, k^*}$ of size at most $d$ be a dominating set of $B_{i^*, j^*, k^*}$. We define a vertex set~$S' \subseteq V(G')$ as follows. Initialize~$S' := \{b_{i^*}, s_{k^*}\}$, and for each vertex~$r_x \in S$, add the corresponding vertex~$r^{j^*}_x \in R_{j^*}$ to~$S'$.

We claim that~$S'$ is a clique in~$G'$. To see this, note that~$R_{j^*}$ is a clique by Step~\ref{step:r}. Vertex~$s_{k^*}$ is adjacent to all vertices of~$R_{j^*}$ by Step~\ref{step:edge:s:r}. Vertex~$b_{i^*}$ is adjacent to all vertices of~$R_{j^*}$ by Step~\ref{step:edge:b:r}. By Step~\ref{step:edge:bs} there is an edge between~$b_{i^*}$ and~$s_{k^*}$. 

Let us consider the weight of clique~$S'$. Since $S$ is a dominating set of $B_{i^*, j^*, k^*}$, if we sum up the weight vectors of the edges~$\{b_{i^*}, r_x^{j^*}\}$ for $r_x \in S$, then by Step~\ref{step:edgewt:b:r} we get a value of at least one at each position of block $k^*$. The most significant position of the resulting sum vector has value $|S| \leq d$. By Step~\ref{step:edge:bs} the weight vector of the edge $\{b_{i^*}, s_{k^*}\}$ consists of all ones, except for block $k^{*}$ and the most significant position, where the value is zero. Thus adding the edge weight of $\{b_{i^*}, s_{k^*}\}$ to the previous sum ensures that each block has value at least~$1$ everywhere, whereas the most significant position has value $|S|$. All other edges spanned by~$S$ have weight~$\vec{0}$. Letting~$t'$ denote the vector obtained by summing the weights of the edges of clique~$S'$, we therefore find that~$t'$ has value~$|S|$ as its most significant position and value at least~$1$ everywhere else.

\bmp{Next we add some additional vertices to the set $S'$ to get a clique of weight exactly $t$. By Step~\ref{step:edge:b,rest}, vertices from the sets~$P_v$ for~$v \in [nz+1]$ are adjacent to all other vertices in the graph and can be added to any clique. All edges incident on a vertex~$p^v_y \in P_v$ have weight~$\vec{0}$, except the edges to vertices of the form~$b_i$ whose weight vector has a~$1$ at the~$v^{th}$ position and~$0$ elsewhere. Since~$S'$ contains exactly one such vertex~$b_{i^*}$, for any~$v \in [nz+1]$ we can add up to~$d-1$ vertices from~$P_v$ to increase the weight sum at position~$v$ from its value of at least~$1$ in~$t'$, to a value of exactly~$d$. Hence~$G'$ has a clique of edge-weight exactly~$t$.}
\end{claimproof}

\begin{claim}
If~$G'$ has a clique of edge-weight exactly~$t$, then some input graph~$G_{i^*, j^*, k^*}$ has a red-blue dominating set of size at most~$d$.
\end{claim}
\begin{claimproof}
Suppose $G'[S']$ is a clique whose total edge weight is exactly $t$. Note that only edges for which one of the endpoints is of the form $b_i$ for $i \in [z]$ have positive edge weights. The remaining edges all have weight~$\vec{0}$. Also, by Step \ref{step:b} there is at most one $b$-vertex in~$S'$. Hence since~$t \neq \vec{0}$ there is exactly one vertex $b_{i^*}$ in~$S'$. By Step~\ref{step:padding} and \ref{step:edge:bp}, the edges of type $\{b_{i^*}, p_y^v\}$ for $p_y^v \in P_v$ contribute at most $d-1$ to the value of each position $v \in [nz+1]$ of the sum. Hence for each position $v \in [nz + 1]$ there is an edge in clique~$S'$ of the form $\{b_{i^*}, r_x^j\}$ or $\{b_{i^*}, s_k\}$ which has a~$1$ at position~$v$. We use this to show there is an input instance with a red-blue dominating set of size at most~$d$.

By Step \ref{step:k}, there is at most one $s$-vertex in~$S'$. Let~$k^* := 1$ if~$S \cap \{s_1, \ldots, s_z\} = \emptyset$, and otherwise let~$s_{k^*}$ be the unique $s$-vertex in~$S'$. Since the weight of the edge~$\{b_{i^*}, s_{k^*}\}$ has zeros in block~$k^*$ by Step~\ref{step:edge:bs}, our previous argument implies that for each of the $n$ positions of block~$k^*$, there is an edge in clique~$S'$ of the form~$\{b_{i^*}, r_x^j\}$ whose weight has a~$1$ at that position. Hence~$S'$ contains at least one $r$-vertex, and by Step~\ref{step:r} all $r$-vertices in the clique~$S'$ are contained in a single set~$R_{j^*}$. We show that~$G_{i^*,j^*,k^*}$ has a red-blue dominating set of size at most~$d$. Let~$S := \{ r_x \mid r_x^{j^*} \in S'\}$. Since for each of the~$n$ positions of block~$k^*$ there is an edge~$\{b_{i^*}, r_x^j\}$ in~$S'$ with a~$1$ at that position, by Step~\ref{step:edge:b:r} each blue vertex of~$B_{i^*,j^*,k^*}$ has a neighbor in~$S$. Hence~$S$ is a red-blue dominating set. By Step~\ref{step:edge:b:r}, the most significant position of each edge incident on $R_{j^*}$ has value~$1$. As the most significant position of the target $t$ is set to~$d$, it follows that~$|S| \leq d$, which proves that~$G_{i^*,j^*,k^*}$ has a red-blue dominating set of size at most~$d$.
\end{claimproof}
This completes the proof of Lemma~\ref{lem:construction:rbds}.
\end{proof}

Lemma~\ref{lem:construction:rbds} forms the main ingredient in a cross-composition that proves kernelization lower bounds for \eewc and its generalization to hypergraphs. For completeness, we formally define the hypergraph version as follows.

\defprob{\eewd(\eewdshort)}{A $d$-uniform hypergraph $G = (V, E)$, weight function $w \colon E(G) \to \mathbb{N}_0$, and a positive integer $t$.}{Does $G$ have a hyperclique of total edge-weight exactly $t$?}

The following theorem generalizes Theorem~\ref{thm:eewc:lb}. \bmp{The case~$d=2$ of the theorem follows almost directly from Lemma~\ref{lem:construction:rbds} and  Theorem~\ref{thm:crosscomp:to:lowerbound}, as the construction in the lemma gives the \mic{crucial} ingredient for a degree-$3$ cross-composition. For larger~$d$, we essentially exploit the fact that increasing the size of hyperedges by one allows one additional dimension of freedom, as has previously been exploited for other kernelization lower bounds for \textsc{$d$-Hitting Set} and~\textsc{$d$-Set Cover}~\cite{DellM12,DellM14}.}
The proof is given in Appendix~\ref{sec:appendix:clique}.

\begin{restatable}{theorem}{lowerboundState} \label{thm:eewd:lb}
$(\bigstar)$ For each fixed $d \geq 2$, \eewd parameterized by the number of vertices $n$ does not admit a generalized kernel of size~$\mathcal{O}(n^{d+1-\varepsilon})$ for any~$\varepsilon > 0$, unless \containment.
\end{restatable}

\subsection{Subset Sum}
\label{sec:ss}
We show that \ssfull parameterized by the number of items $n$ does not have generalized kernel of bitsize $\Oh(n^{2-\varepsilon})$ for any $\varepsilon > 0$, unless \containment. We prove the lower bound by giving a linear-parameter transformation from \erbds. \bmp{We use \erbds rather than \rbds as our starting problem for this lower bound because it will simplify the construction: it will avoid the need for `padding' to cope with the fact that vertices are dominated multiple times.}

The \ssfull problem is formally defined as follows.

\defparprob{\textsc{Subset Sum (SS)}}{A multiset~$X$ of $n$ positive integers and a positive integer $t$.}{$n$}{Does there exist a subset $S \subseteq X$ with $\sum_{x \in S} x = t$?}

We use the following problem as the starting point of the reduction.

\defparprob{\textsc{Exact Red-Blue Dominating Set (ERBDS)}}{A bipartite graph $G$ with a bipartition of~$V(G)$ into sets~$R$ (red vertices) and~$B$ (blue vertices), and a positive integer~$d \leq |R|$.}{$n := |V(G)|$}{Does there exist a set $D \subseteq R$ of size \emph{exactly} $d$ such that every vertex in $B$ has exactly one neighbor in $D$?}

Jansen and Pieterse proved the following lower bound for \erbdsshort.
\begin{theorem}[{\cite[Thm. 4.9]{JansenP19}}] \label{thm:erbds:lb}
\textsc{Exact Red-Blue Dominating Set} parameterized by the number of vertices~$n$ does not admit a generalized kernel of size~$\mathcal{O}(n^{2-\varepsilon})$ unless \containment.
\end{theorem}

Actually, the lower bound they proved is for a slightly different variant of \erbdsshort where the solution $D$ is required to have size \emph{at most $d$}, instead of exactly~$d$. Observe that the variant where we demand a solution of size \emph{exactly} $d$ is at least as hard as the \emph{at most} $d$ version: the latter reduces to the former by inserting~$d$ isolated red vertices.
Therefore the lower bound by Jansen and Pieterse also works for the version we use here, which will simplify the presentation. 

\subsetSumState*
\begin{proof}
Given a graph~$G$ with a bipartition of $V(G)$ into $R$ and $B$ with $R = \{r_1, r_2, \ldots , r_{n_R}\}$, $B = \{b_1, b_2, \ldots , b_{n_B}\}$, and target value~$d$ for \erbdsshort, we transform it to an equivalent instance $(X, t)$ of \ssshort such that $|X| = n_R$. We start by defining $n_R$ numbers $N_1, N_2, \ldots, N_{n_R}$ in base $(n_R + 1)$. For each~$i \in [n_R]$, the number $N_i$ consists of $(n_B + 1)$ digits. We denote the digits of the number~$N_i$ by $N_i[1], \ldots, N_i[n_B + 1]$, where $N_i[1]$ is the least significant and $N_i[n_B + 1]$ is the most significant digit. Intuitively, the number~$N_i$ corresponds to the red vertex~$r_i$. See Figure~\ref{fig:subsetsum} for an illustration.

For each~$i \in [n_R]$, for each $j \in [n_B + 1]$, digit $N_i[j]$ of number $N_i$ is defined as follows:
\begin{equation}
N_i[j] = \begin{cases}
1 & \mbox{if $j = n_B +1$} \\
1 & \mbox{if $j \in [n_B]$ and $\{r_i, b_j\} \in E(G)$} \\
0 & \mbox{otherwise.}
\end{cases}
\end{equation}
\bmp{Hence the most significant digit of each number is~$1$, and the remaining digits of number~$N_i$ form the $0/1$-vector indicating to which of the~$n_B$ blue vertices~$r_i$ is adjacent in~$G$.}

To complete the construction we set $X = \{N_1, N_2, \ldots, N_{n_R}\}$ and we define $t$ as follows:
\begin{equation}
t = d \underbrace{11\ldots 1}_{n_B \text{ times}}
\end{equation}
\bmp{Observe that under these definitions, there are no carries when adding up a subset of the numbers in~$X$, as each digit of each of the~$n_R$ numbers is either~$0$ or~$1$ and we work in base~$n_R+1$.}

\bmp{The number of items~$|X|$ in the constructed instance of \ssshort is~$n_B$, linear in the parameter~$|V(G)|$ of \erbdsshort. It is easy to see that the construction can be carried out in polynomial time. To complete the linear-parameter transformation from \erbdsshort to \ssshort, it remains to prove that}~$G$ has a set~$D \subseteq R$ of size exactly~$d$ such that every vertex in~$B$ has exactly one neighbor in~$D$, if and only if there exist a set~$S \subseteq X$ with $\sum_{x \in S} x = t$.

\begin{figure}
    \centering
    \includegraphics[scale=0.8]{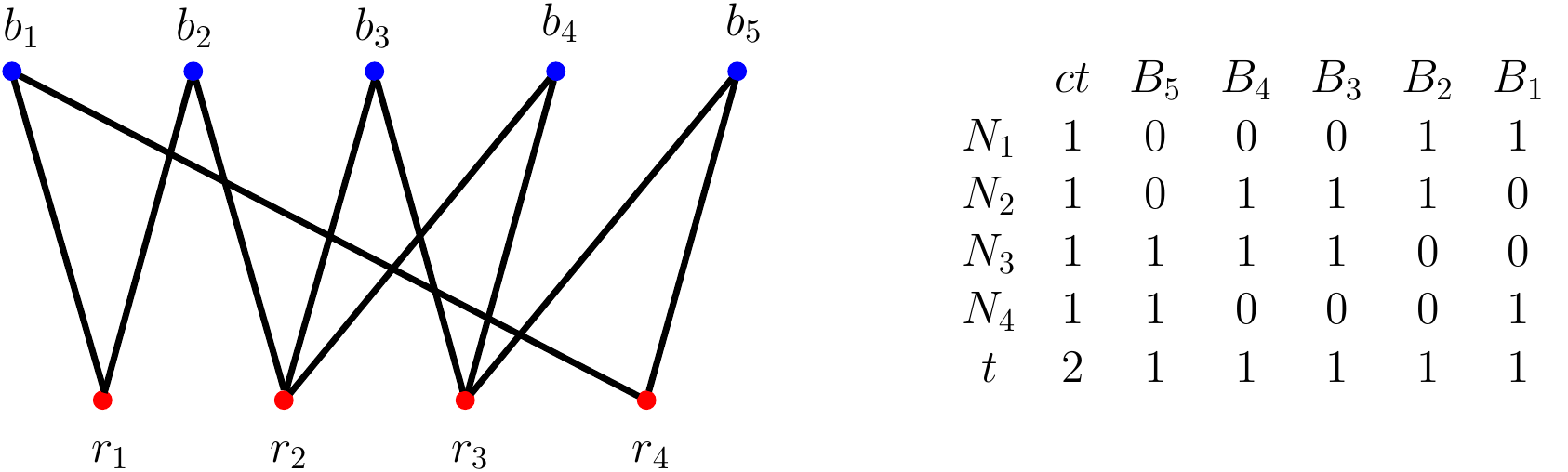}
    \caption{Left: An instance of \erbdsshort with $n_R = 4$, $n_B= 5$, and $d=2$. Right: Illustration of the \ssshort instance created for the given input. Note that $\{r_2, r_4\}$ and the numbers $\{N_2, N_4\}$ form a solution for \erbdsshort and \ssshort, respectively. The leftmost column corresponds to the total count ($ct$) of the number of elements; the remaining columns correspond to blue vertices.}
    \label{fig:subsetsum}
\end{figure}

In the forward direction, suppose that there exists a set $D \subseteq R$ of size exactly $d$ such that every vertex in $B$ has exactly one neighbor in $D$. \bmp{We claim that~$\{ N_i \mid r_i \in D\}$ is a solution to \ssshort. The resulting sum has value~$d$ at the most significant digit since~$|D| = d$. All other digits correspond to vertices in~$B$. Since each blue vertex is adjacent to exactly one vertex from~$D$ it is easy to verify that all remaining digits of the sum are exactly one, implying that the numbers sum to exactly~$t$.}

For the reverse direction, suppose there is a set~$S \subseteq X$ with $\sum_{x \in S} x = t$. 
\bmp{Since the most significant digit of~$t$ is set to~$d$ and each number in~$X$ has a~$1$ as most significant digit, we have~$|S| = d$ since there are no carries during addition.} 
Define~$D := \{ r_i \mid N_i \in S \}$ as the set of the red vertices corresponding to the numbers in~$S$. As~$\sum _{x \in S} x = t$ and no carries occur in the summation, we have $\sum _{x \in S} x[j] = t[j] = 1$ for each~$j \in [n_B]$. As the $j$-th digit of all numbers is either $0$ or $1$ by definition, there is a unique $N_i \in S$ with $N_i[j] = 1$, so that $r_i \in D$ is the unique neighbor of $b_j$ in $D$. This shows that $D$ is an exact red-blue dominating set of size $d$, concluding the linear-parameter transformation.

\bmp{
If there was a generalized kernelization for \ssshort of size~$\Oh(n^{2-\varepsilon})$, then we would obtain a generalized kernelization for \erbdsshort of size~$\Oh(n^{2-\varepsilon})$ by first transforming it to \ssshort, incurring only a constant-factor increase in the parameter, and then applying the generalized kernelization for the latter. Hence by contraposition and Theorem~\ref{thm:erbds:lb}, the claim follows.}
\end{proof}

\subsection{Constraint Satisfaction Problems}\label{sec:csp}

In this section we extend our lower bounds to cover Boolean Constraint Satisfaction Problems (CSPs).
We employ the recently introduced framework~\cite{JansenW20} of reductions among different CSPs to make a connection with \eewdshort.
We start with introducing terminology necessary to identify crucial properties of CSPs.

\subparagraph*{Preliminaries on CSPs}
A~$k$-ary constraint is a~function $f 
\colon \{0,1\}^k \rightarrow \{0,1\}$.
We refer to $k$ as the~arity of~$f$, denoted $\ar(f)$.
We always assume that the~domain is Boolean.
A~constraint $f$ is satisfied by an~input $s \in \{0,1\}^k$ if $f(s) = 1$.
A~constraint language
$\Gamma$ is a~finite collection of~constraints $\{f_1, f_2, 
\dots, f_\ell\}$, potentially with different arities.
{A~\emph{constraint application}, of~a~$k$-ary constraint $f$ to a~set of~$n$ Boolean~variables, is a~triple $\langle f, (i_1, i_2, \dots i_k), w \rangle$, where
the~indices $i_j \in [n]$ select $k$ of~the~$n$ Boolean~variables to whom the~constraint is applied, and $w$ is an integer~weight.}
The variables can~repeat in a~single application.

A~formula $\Phi$ of~\cspg is a~set of~constraint applications from $\Gamma$ over a~{common} set of~variables. 
For an~assignment $x$, that is, a mapping
from the set of variables to $\{0, 1\}$, the~integer $\Phi(x)$ is the~sum of~weights of~the~constraint applications satisfied by $x$.
The considered decision problems are defined as follows.

\defparprob{\ewcspg}
{A formula $\Phi$ of \cspg over $n$ variables, an integer $t\in \zz$.}
{$n$}
{Is there an assignment  $x$ for which $\Phi(x) = t$?}

\defparprob{\mwcspg}
{A formula $\Phi$ of \cspg over $n$ variables, an integer $t\in \zz$.}
{$n$}
{Is there an assignment $x$ for which $\Phi(x) \geq t$?}

\depr{
\defparprob{Weighted Set-Exact CSP$(\Gamma, \mathbb{W})$}
{A formula $\Phi$ of $CS(\Gamma, \mathbb{W})$ over $n$ variables, set $S \subseteq \mathbb{W}$}
{$n$}
{Is there an assignment vector $x$ for which $\Phi(x) \in S$?}
}

The compressibility of \mwcspg has been studied by Jansen and Włodarczyk~\cite{JansenW20}, who obtained essentially optimal kernel sizes for every $\Gamma$ in the case where the weights are polynomial with respect to $n$.
Even though the upper and lower bounds in~\cite{JansenW20} are formulated for \mwcspg, they could be adapted to work with \ewcspg.
The crucial idea which allows to determine compressibility of $\Gamma$ is the representation of constraints via multilinear polynomials.

\begin{definition}\label{def:characteristic-poly}
For a~$k$-ary constraint $f \colon \{0,1\}^k \to \{0,1\}$ its characteristic polynomial $P_f$ is the~unique $k$-ary multilinear polynomial over $\mathbb{R}$ satisfying $f(x) = P_f(x)$ for any $x \in \{0,1\}^k$.
\end{definition}

It is known that such a polynomial always exists and it is unique~\cite{NisanS94}.

\begin{definition}\label{def:deg-gamma}
The degree of constraint language $\Gamma$, denoted $\deg(\Gamma)$, is the maximal degree of a~characteristic polynomial $P_f$ \bmp{over all} $f \in \Gamma$.
\end{definition}

The main result \bmp{of Jansen and Włodarczyk}~\cite{JansenW20} states that \mwcspg with polynomial weights admits a kernel of $\Oh(n^{\deg(\Gamma)} \log n)$ \bmp{bits} and, as long as the problem is NP-hard, it does not admit a kernel of size $\Oh(n^{\deg(\Gamma) - \eps})$, for any $\eps > 0$, unless \containment.
It turns out that in the variant when we allow both positive and negative weights the problem is NP-hard whenever $\deg(\Gamma) \ge 2$~\cite{csp-nphard}.
The lower bounds are obtained via linear-parameter transformations, where the parameter is the number of variables $n$.
We shall take advantage of the fact that these transformations still work for an~unbounded range of weights.

\begin{lemma}[\cite{JansenW20}, Lemma 5.4]\label{thm:deg1}
For constraint languages $\Gamma_1, \Gamma_2$ such that $2 \le \deg(\Gamma_1) \le \deg(\Gamma_2)$, there is a polynomial-time algorithm that, given a formula $\Phi_1 \in \csp(\Gamma_1)$ on $n_1$ variables and integer $t_1$, returns a formula
$\Phi_2 \in \csp(\Gamma_2)$ on $n_2$ variables and integer $t_2$, such that
\begin{enumerate}
    \item $n_2 = \Oh(n_1)$,
    \item $\exists_x \Phi_1(x) = t_1 \Longleftrightarrow \exists_y \Phi_2(y) = t_2$,
    \item $\exists_x \Phi_1(x) \ge t_1 \Longleftrightarrow \exists_y \Phi_2(y) \ge t_2$.
\end{enumerate}
\end{lemma}

\subparagraph*{Kernel lower bounds for CSP}

The lower bound of \bmp{$\Omega(n^{\deg(\Gamma) - \eps})$} has been obtained via a reduction from $d$-SAT (with $d = \deg(\Gamma)$) to \mwcspg, combined with the fact that \textsc{Max $d$-SAT} does not admit a kernel of size $\Oh(n^{d - \eps})$ for $d \ge 2$~\cite{DellM14, JansenW20}.
We are going to show that when the weights are arbitrarily large, then 
the optimal compression size for \ewcspg becomes essentially $\Oh(n^{\deg(\Gamma) + 1})$, so the exponent is always larger by one compared to the case with polynomial weights. 
To this end, we are going to combine the aforementioned reduction framework with our lower bound for \eewd.

Consider a constraint language $\Gamma_{\textsc{and}}^d$ consisting of a single $d$-ary constraint $\textsc{AND}_d$, which is satisfied only if all the arguments equal 1.
The characteristic polynomial of $\textsc{AND}_d$ is simply $P(x_1, \dots, x_d) = x_1x_2\cdots x_d$, hence the degree of $\Gamma_{\textsc{and}}^d$ equals $d$.
We first translate our lower bounds for the hyperclique problems into a lower bound for \ewcsp$(\Gamma_{\textsc{and}}^d)$ for all $d \ge 2$, and then extend it to other CSPs.

\begin{restatable}{lemma}{cspGammaLB}
\label{lem:gamma_and_lb}$(\bigstar)$
For all $d \ge 2$, \ewcsp$(\Gamma_{\textsc{and}}^d)$
does not admit a generalized kernel of size $\Oh(n^{d + 1 - \eps})$, for any $\eps > 0$, unless
\containment.
\end{restatable}

\bmp{The lower bound for~\ewcsp$(\Gamma_{\textsc{and}}^d)$ given by Lemma~\ref{lem:gamma_and_lb} yields a lower bound for general~\ewcsp$(\Gamma)$ using the reduction framework described above.}

\begin{theorem}\label{lem:gamma_general_lb}
For any~$\Gamma$ with~$\deg(\Gamma) \ge 2$, \ewcsp$(\Gamma)$
does not admit a generalized kernel of size $\Oh(n^{\deg(\Gamma) + 1 - \eps})$, for any $\eps > 0$, unless
\containment.
\end{theorem}
\begin{proof}
Consider an $n$-variable instance $(\Phi_1, t_1)$ of \textsc{Weighted Exact} \csp$(\Gamma_{\textsc{and}}^d)$, where $d = \deg(\Gamma)$.
It holds that $\deg(\Gamma_{\textsc{and}}^d) = d$.
By Lemma~\ref{thm:deg1}, there is a linear-parameter transformation that \bmp{translates}  
$(\Phi_1, t_1)$ into an equivalent instance $(\Phi_2, t_2)$ of \textsc{Weighted Exact} \cspg.
If we could compress $(\Phi_2, t_2)$ into $\Oh(n^{d + 1 - \eps})$ bits, this would entail the same compression for $(\Phi_1, t_1)$.
The claim follows from Lemma~\ref{lem:gamma_and_lb}.
\end{proof}

\bmp{This concludes the discussion of kernelization lower bounds. The kernelization upper bounds discussed in the introduction can be found in Appendix~\ref{sec:ub}.} 

\section{Node-weighted Vertex Cover in bipartite graphs}
\label{sec:vc:main}

\subparagraph*{Preserving all minimum solutions}

For a graph~$G$ with node-weight function~$w \colon V(G) \to \mathbb{N}_+$, we denote by~$\mathcal{C}(G,w)$ the collection of subsets of~$V(G)$ which are minimum-weight vertex covers of~$G$. For $n$-vertex \emph{bipartite} graphs there exists a weight function with range~$[n]$ that preserves the set of minimum-weight vertex covers, which can be computed efficiently.

\begin{restatable}{theorem}{bipartiteWeightStatement}
\label{theorem:bipartiteweightreduction}$(\bigstar)$
There is an algorithm that, given an $n$-vertex bipartite graph~$G$ and node-weight function~$w \colon V(G) \to \mathbb{N}_+$, outputs a weight function~$w^*\colon V(G) \to [n]$ such that $\mathcal{C}(G, w) = \mathcal{C}(G, w^*)$. The running time of the algorithm is polynomial in~$|V(G)|$ and the binary encoding size of~$w$.
\end{restatable}

The proof of the theorem is given in Appendix~\ref{sec:vc}. It relies on the fact that a maximum $b$-matching (the linear-programming dual to \textsc{Vertex Cover}) can be computed in strongly polynomial time in bipartite graphs by a reduction to \textsc{Max Flow}. The structure of a maximum $b$-matching allows two weight-reduction rules to be formulated whose exhaustive application yields the desired weight function. The bound of~$n$ on the largest weight in Theorem~\ref{theorem:bipartiteweightreduction} is best-possible, which we prove in Lemma~\ref{lemma:weight_reduction_is_tight} in Appendix~\ref{sec:vc}. 

\subparagraph*{Preserving the relative weight of solutions}
For a graph~$G$, we say that two node-weight functions~$w,w'$ are \emph{vertex-cover equivalent} if the ordering of inclusion-minimal vertex covers by total weight is identical under the two weight functions, i.e., for all pairs of inclusion-minimal vertex covers~$S_1, S_2 \subseteq V(G)$ we have~$w(S_1) \leq w(S_2) \Leftrightarrow w'(S_1) \leq w'(S_2)$. While a minimum-weight vertex cover of a bipartite graph can be found efficiently, the following theorem shows that nevertheless weight functions with exponentially large coefficients may be needed to preserve the ordering of minimal vertex covers by weight.

\begin{restatable}{theorem}{vcThresholdCountState}
\label{theorem:vclowerbound2}$(\bigstar)$
For each~$n \geq 1$, there exists a node-weighted bipartite graph~$G_n$ on~$2(n+1)$ vertices with weight function $w \colon V(G_n) \to \mathbb{N}_+$ such that for \sk{all weight functions~$w' \colon V(G) \to \mathbb{N}_+$ which are vertex-cover equivalent to~$w$, we have:} $\max\limits_{v \in V(G_n)} w'(v) \geq 2^{\Omega(n)}$.
\end{restatable}

\section{Conclusions}
We have established kernelization lower bounds for \ssfull, \eewd, and a family of \ewcsp\bmp{ problems}, which
make it unlikely \bmp{that there exists an efficient algorithm to} compress a single weight into~$o(n)$ bits. 
This gives a clear separation between the setting involving arbitrarily large weights and the case with polynomially-bounded weights, which can be encoded with $\Oh(\log n)$ bits each.
The matching kernel upper bounds are randomized and we leave it as an open question to derandomize them.
For \ssfull parameterized\bmp{ by the number of \mic{items}~$n$,}  a~deterministic kernel of size $\Oh(n^4)$ is known~\cite{EtscheidKMR17}.

Kernelization of minimization/maximization problems is so far less understood.
We are able to match the same kernel size as for the exact-weight problems, but only through Turing kernels.
Using techniques from~\cite{EtscheidKMR17} one can obtain, e.g., a kernel of size $\Oh(n^8)$ for \textsc{Max-Edge-Weight Clique}.
Improving upon this bound possibly requires a~better understanding of the threshold functions.
Our study of weighted \textsc{Vertex Cover} on bipartite graphs indicates that preserving the order between all the solutions might be overly demanding and it could be easier to keep track only of the structure of the optimal solutions.
Can we extend the theory of threshold functions so that better bounds are feasible when we just want to maintain a~separation between optimal and non-optimal solutions?

\bibliography{bibliography.bib}

\appendix

\clearpage

\section{Kernel lower bounds}
\label{sec:klb:app}

\subsection{Omitted proofs for \eewc} \label{sec:appendix:clique}

\lowerboundState*
\begin{proof}
We give a degree-$(d+1)$ cross-composition (Definition~\ref{def:deg-dcross-comp}) from \rbdsshort to the weighted hyperclique problem using Lemma~\ref{lem:construction:rbds}. We start by giving a polynomial equivalence relation $\mathcal{R}$ on inputs of \rbdsshort. Let two instances of \rbdsshort be equivalent under $\mathcal{R}$ if they have the same number of red vertices, the same number of blue vertices, and the same target value $d$. It is easy to check that $\mathcal{R}$ is a polynomial equivalence relation.

Consider $Z$ inputs of \rbdsshort from the same equivalence class of $\mathcal{R}$. If~$Z$ is not a~$(d+1)^{th}$ power of an integer, then we duplicate one of the input instances until we reach the first number of the form~$2^{(d+1)i}$, which is trivially such a power. This increases the number of instances by at most the constant factor~$2^{d+1}$ and does not change whether there is a YES-instance among the instances. As all requirements on a cross-composition are oblivious to constant factors, from now on we may assume without loss of generality that~$Z = z^{d+1}$ for some integer~$z$. By definition of~$\mathcal{R}$, all instances have the same number~$m$ of red vertices, the same number~$n$ of blue vertices, and have the same maximum size~$d$ of a solution.

\bmp{For~$d=2$, we can simply invoke Lemma~\ref{lem:construction:rbds} for the~$z^{d+1} = z^3$ instances of \rbdsshort and output the resulting instance~$(G', w, t)$ of \eewcshort, which acts as the logical OR. Since the encoding size~$N$ of an instance of \rbdsshort with~$m$ red vertices,~$n$ blue vertices, and target value~$d$ satisfies~$N \in \Omega(n + m + d)$, Lemma~\ref{lem:construction:rbds} guarantees that~$G'$ has~$\Oh(z(m + nd)) \in \Oh(\sqrt{Z} \cdot N^2)$ vertices, which is suitably bounded for a degree-2 cross-composition for the parameterization by the number of vertices. Hence the claimed lower bound for generalized kernelization then follows from Theorem~\ref{thm:crosscomp:to:lowerbound}.} 

\bmp{In the remainder of the proof, we assume~$d \geq 3$.} Partition the~$z^{d+1}$ inputs in $z^{d-2}$ groups $\{ X_{i_1, \ldots, i_{d-2}} \mid i_1, \ldots, i_{d-2}\}$ of size $z^3$ each. Apply Lemma~\ref{lem:construction:rbds} to each group~$X_{i_1, \ldots, i_{d-2}}$.  This results in $z^{d-2}$ instances $(G_{i_1, \ldots, i_{d-2}}, w_{i_1, \ldots, w_{d-2}}, t)$ of \eewcshort on a simple graph. Note that all instances share the same value of~$t > 0$, as Lemma~\ref{lem:construction:rbds} ensures that~$t$ only depends on~$(z,d,n,m)$ which are identical for all groups. Similarly, all resulting instances have the same number of vertices. Hence we can re-label the vertices in each graph so that all graphs~$G_{i_1, \ldots, i_{d-2}}$ have the same vertex set~$\mathcal{V}$ of size~$\Oh(z (m + nd))$. The YES/NO-answer to each composed instance is the disjunction of the answers to the \rbdsshort instances in its corresponding group. 

Build a $d$-uniform hypergraph $G^*$ with weight function~$w^* \colon E(G^*) \to \mathbb{N}_0$ and target value~$t^*$ as follows:
\begin{enumerate}
    \item $V(G^*)= \mathcal{V} \cup Y_1 \cup \cdots \cup Y_{d-2}$, where $Y_\ell=\{y_{\ell,j} \mid j \in [r]\}$ for $\ell \in [d-2]$. \label{hypergraph:vertices}
    
    \item A set~$S \subseteq V(G^*)$ of exactly $d$ vertices is a hyperedge of $G^*$ if there is no~$\ell \in [d-2]$ for which~$|S \cap Y_\ell| > 1$.\label{hypergraph:edge:augmentation}
    
    \item The weight of a hyperedge~$S$ is equal to~0 if there exists~$\ell \in [d-2]$ with $S\cap Y_\ell = \emptyset$. Otherwise, for each~$\ell \in [d-2]$ let~$i_\ell$ be the unique index~$j$ such that~$y_{\ell,j} \in S$. 
    \begin{itemize}
        \item If~$e_S := S \cap \mathcal{V}$ is an edge in graph~$G_{i_1, \ldots, i_{d-2}}$ then define~$w^*(S) := w_{i_1, \ldots, i_{d-2}}(e_S)$. 
        \item Otherwise, let~$w^*(S) := t+1$.
    \end{itemize}
    \label{hypergraph:edge:weights}
    
    \item Set~$t^*=t$. \label{hypergraph:target}
\end{enumerate}

Since~$d \in \Oh(1)$, hypergraph $G^*$ has $\Oh(z \cdot (m+nd)) + \mathcal{O}(z \cdot d) \in \Oh(Z^{1/(d+1)}\cdot (m+n)^{\mathcal{O}(1)})$ vertices. (We use here that~$d \leq m$.) Hence the parameter value of the constructed \eewd instance is indeed bounded by the $(d+1)$-th root of the number of input instances times a polynomial in the maximum size of an input instance, satisfying the parameter bound of a degree-$(d+1)$ cross-composition.

It remains to verify that $G^*$ has a hyperclique of weight $t^*$ if and only if one of the input instances has a \rbdsshort of size at most $d$. By the guarantee of Lemma~\ref{lem:construction:rbds}, it suffices to show that $G^*$ has a hyperclique of weight $t^*$ if and only if one of the weighted standard graphs~$(G_{i_1, \ldots, i_{d-2}}, w_{i_1, \ldots, i_{d-2}})$ obtained by applying that lemma to some group of~$z^2$ inputs, has a clique of weight~$t$.

First suppose there exists a weighted graph~$(G_{i^*_1, \cdots, i^*_{d-2}}, w_{i^*_1, \ldots, i^*_{d-2}})$ that contains a clique~$S$ of total edge weight $t$. Let~$I := \{ y_{\ell, i^*_\ell} \mid \ell \in [d-2]\}$. Let~$S' := S \cup I$. By Step~\ref{hypergraph:edge:augmentation}, the set~$S'$ is a hyperclique in~$G^*$. It remains to verify that its weight is~$t^* = t$. By Step~\ref{hypergraph:edge:weights}, for each edge~$e$ of the clique~$S$ the set~$e \cup I$ is a hyperedge in~$G^*$ of the same weight. Additionally, each subset of~$S'$ that does not contain~$I$ has weight~$0$. Hence the weight of hyperclique~$S'$ is equal to the weight of clique~$S$ and is therefore~$t^* = t$.

For the other direction, suppose~$G^*$ has a clique~$G^*[S^*]$ of weight~$t^* = t$. Since~$t > 0$ and all hyperedges in~$G^*$ of nonzero weight contain exactly one vertex of each set~$Y_\ell$ for~$\ell \in [d-2]$, there exist~$i^*_1, \ldots, i^*_{d-2}$ such that~$S^* \cap Y_\ell = \{i^*_\ell\}$ for each~$\ell \in [d-2]$. Let~$I := \{y_{\ell, i^*_\ell} \mid \ell \in [d-2]\}$. We will show that~$S^* \cap \mathcal{V}$ is a clique of weight~$t$ in~$G_{i^*_1, \ldots, i^*_{d-2}}$. Since~$t^* = t > 0$ and edge-weights are non-negative, it follows that no hyperedge in~$S^*$ has weight~$t+1$. By Step~\ref{hypergraph:edge:weights}, this implies each subset of~$S^* \cap \mathcal{V}$ of size two is an edge of~$G_{i^*_1, \ldots, i^*_{d-2}}$, and hence~$S^* \cap \mathcal{V}$ is a clique. For each set~$e \subseteq \mathcal{V}$ of size two, the weight of the hyperedge~$e \cup I$ is equal to~$w_{i^*_1, \ldots, i^*_{d-2}}(e)$. As all other hyperedges in~$S^*$ have weight~$0$, it follows that the weight of the clique~$S^* \cap \mathcal{V}$ equals that of hyperclique~$S^*$, and is therefore equal to~$t^* = t$. This implies $G_{i^*_1, \cdots, i^*_{d-2}}$ has a clique of total edge weight~$t=t^*$, which concludes the proof.
\end{proof}

\subsection{Omitted proofs for CSPs}

\cspGammaLB*
\begin{proof}
Consider an instance $(G,w,t)$ of \eewd.
Let $W$ be the sum of all weights, which are by the definition non-negative.
We can assume $t \in [0, W]$, as otherwise there is clearly no solution.
We create an instance $\Phi$ of \ewcsp$(\Gamma_{\textsc{and}}^d)$ with the variable set $V(G)$ as follows.
For each potential hyperedge $e = \{v_1, \dots, v_d\}$, if $e \in E(G)$ we create a constraint application $\langle \textsc{AND}_d, (v_1, \dots, v_d), \bmp{w(e)} \rangle$ and if $e \not\in E(G)$, we create a constraint application $\langle \textsc{AND}_d, (v_1, \dots, v_d), W+1 \rangle$.

If $X \subseteq V(G)$ is a hyperclique with total weight $t$, then for the assignment $x(v) = [v \in X]$ it holds that $\Phi(x) = t$.
In the other direction, if $\Phi(x) = t$ then $x$ cannot satisfy any constraint application with weight $W+1$.
Hence, each size-$d$ subset of 1-valued variables corresponds to \bmp{a} hyperedge in $G$ and $X = \{v \in V(G) \mid x(v) = 1\}$ forms a hyperclique of total weight $t$.

We have constructed a linear-parameter transformation from \eewdshort to \ewcsp$(\Gamma_{\textsc{and}}^d)$.
Therefore, any generalized kernel of size $\Oh(n^{d + 1 - \eps})$ for the latter would entail the same bound for \eewdshort. 
The claim follows from Theorem~\ref{thm:eewd:lb}.
\end{proof}

\section{Kernel upper bounds}
\label{sec:ub}

\depr{
\begin{theorem}[\cite{FrankT87}]
There is an algorithm that,
given a vector $w \in \mathbb{Q}^r$ and an integer $N$,
in polynomial time finds a vector $w'\in \mathbb{Z}^r$
with $||w'||_\infty \le 2^{4r^3}N^{r(r+1)}$
such that $\text{sign}(w \cdot b) = \text{sign}(w' \cdot b)$
for all vectors $b \in \mathbb{Z}^r$
with $||b||_1 \le N-1$.
\end{theorem}

\begin{corollary}[\cite{EtscheidKMR17}]\label{cor:threshold-reduction}
There is an algorithm that, given a vector $w \in \mathbb{Q}^r$ and a rational $W \in \mathbb{Q}$, in polynomial time finds a vector $w'\in \mathbb{Z}^r$ with $||w'||_\infty \le 2^{\Oh(r^3)}$ an integer $W'$
such that $\sum_{x \in X} w(x) \le W \Longleftrightarrow \sum_{x \in X} w'(x) \le W'$ for all subsets $X \subseteq [r]$.
\end{corollary}
}

In this section, we present randomized kernel upper bounds for \eewdshort and the considered family of \csp, which match the obtained lower bounds.
For the maximization variant of \eewdshort,
we present a Turing kernel with the same bounds.
The results in this section follow from combining known arguments from Harnik and Naor~\cite{HarnikN10} and Nederlof et al.~\cite{Nederlof20} with gadgets that allow us to produce an instance of the same problem that is being compressed (so we obtain a true kernelization, not a generalized one).

Consider a family $\mathcal{F}$ of subsets of \bmp{a} universe $U$ and a weight function $w \colon \mathcal{F} \rightarrow [-N, N]$.
For a subset $X \subseteq U$, we denote $\wsum{w}(X) = \sum_{Y \in \mathcal{F},\, Y \subseteq X} w(Y)$.
The following fact has been observed by Harnik and Naor~\cite[Claim 2.7]{HarnikN10} and for the sake of completeness we provide a~proof for the formulation which is the most convenient for us.

\begin{lemma}\label{lem:prime-hashing}
Let $U$ be a set of size $n$, $\mathcal{F} \subseteq 2^U$ be a family of subsets, $w \colon \mathcal{F} \rightarrow [-N, N]$ be a weight function, and $t \in  [-N, N]$.
There exists a randomized polynomial-time algorithm that, given a real $\eps > 0$, returns 
a~prime number $p \le 2^n \cdot \mathrm{poly}(n, \log N, \eps^{-1})$,
such that if there is no $X \subseteq U$ satisfying $\wsum{w}(X) = t$, then
\[\mathbb{P}\bigg(\text{there is }X \subseteq U \text{ satisfying } \wsum{w}(X) \equiv t \pmod {p} \bigg) \le \eps.
\]
\end{lemma}
\begin{proof}
For \bmp{a} fixed function $w$, we say that $p$ is \emph{bad} if for some $X \subseteq U$ it holds that
$\wsum{w}(X) \equiv t \pmod {p}$ but $\wsum{w}(X) \neq t$.
This implies that $p$ divides $|\wsum{w}(X) - t|$.
We argue that the number of bad primes is bounded by $2^n\cdot (n + 1 + \log(N))$.
Since $|\wsum{w}(X) - t| \le 2^{n+1}\cdot N$, this number can have at most $\log(2^{n+1}\cdot N) = n + 1 + \log{N}$ different prime divisors. 
There are at most $2^n$ choices of $X$, which proves the bound. 

We sample a random prime $p$ among the set of the first $M = 2^n\cdot (n + 1 + \log(N)) \cdot \eps^{-1}$ primes.
It is known that \bmp{the first~$M$ primes lie} in the interval $[2, \Oh(M \log M)]$
and we can uniformly sample a prime number from this interval in time $(\log M)^{\Oh(1)} = (n+\log\log(N)+\log(\eps^{-1}))^{\Oh(1)}$~\cite{karp1987efficient}.
By the argument above, the probability of choosing a bad prime is bounded by $\eps$.
\end{proof}

\depr{
\begin{lemma}[\cite{HarnikN10, NederlofLZ12}]
Let $U$ be a set of size $n$.
There exists a randomized polynomial-time algorithm that, given $w \colon U \rightarrow [N]$, integer $t$, and a real $\eps > 0$, returns  $w' \colon U \rightarrow [M]$ and $S \subseteq [M]$, where $M \le 2^n \cdot \text{poly}(n, \log N, |S|, \eps^{-1})$,
such that for every set family $\mathcal{F} \subseteq 2^U$:
\begin{enumerate}
    \item for any $X \in \mathcal{F}$ such that $w(X) = t$, it holds $w'(X) \in S$.
    \item if there is no $X \in \mathcal{F}$ such that $w(X) \in S$, then \newline $\prob{\text{there is }X \in \mathcal{F} \text{ such that } w'(X) \in S} \le \eps$.
\end{enumerate}
\end{lemma}
}

\begin{theorem}\label{thm:eewd-upper-bound}
There is a randomized polynomial-time algorithm that, given an $n$-vertex instance $(G,w,t)$ of \eewd,
outputs an instance $(G',w',t')$ of bitsize $\Oh(n^{d+1})$, such that:
\begin{enumerate}
    \item if $(G,w,t)$ is a \textsc{yes}-instance, then $(G',w',t')$ is always a \textsc{yes}-instance,
    \item if $(G,w,t)$ is a \textsc{no}-instance, then $(G',w',t')$ is a \textsc{no}-instance with probability at least $1 - 2^{-n}$.
\end{enumerate}
Furthermore, each number in $(G',w',t')$ is bounded by $2^{\Oh(n)}$.
\end{theorem}
\begin{proof}
Let us
define $N = \max(t,\, \max_{e \in E(G)} w_e)$. 
We can assume $\log N \le 2^n$, because otherwise the input length is lower bounded by $2^n$ and the brute-force algorithm for \eewdshort becomes polynomial.

We apply Lemma~\ref{lem:prime-hashing} to the weight function $w$, target $t$, and $\eps = 2^{-n}$, to compute the desired prime~$p \le 2^n \cdot \text{poly}(n, \log N, \eps^{-1}) = 2^{\Oh(n)}$.
If there exists a hyperclique $X \subseteq V$ satisfying $\wsum{w}(X) = t$ \bmp{with respect to the weighted set family~$E(G) \subseteq \binom{V(G)}{d}$}, then clearly $\wsum{w}(X) \equiv t \pmod{p}$.
Furthermore, with probability $1-2^{-n}$, the implication in the other direction holds as well.
In particular, in this case $p$ does not divide $t$. 

Let us construct a new instance \bmp{$(G', w', t')$} of \eewd with weights bounded by $p\cdot n^d$, which is bounded by $2^{\Oh(n)}$ for constant $d$.
We set $w'(v) = w(v) \pmod p$ and $t_p = t \pmod p$.
The condition $\wsum{w}(X) \equiv t \pmod{p}$ is equivalent to \bmp{the existence of $i \in [0, n^d)$ for which $\wsum{w'}(X) = t_p + ip$}, because the sum $\wsum{w'}(X)$ comprises of at most $n^d$ summands from the range $[0,p)$.

We introduce a set $U_Z$ of $d-1$ new vertices and for each $j \in [0,d-1]$ we introduce a set $U_j$ of $n$ new vertices. \bmp{Intuitively, the sets~$U_j$ can be used to represent any number~$i \in [0, n^d)$ in base~$n$.}
For every \bmp{$j$ and every} $v \in U_j$ we create a hyperedge $e = U_Z \cup \{v\}$ with weight $w'_e = n^j\cdot p$.
For every other size-$d$ subset containing at least one new vertex, we create a hyperedge with weight 0.
Observe that for every integer $i \in [0, n^d]$, we can find a set $Y \subseteq U_Z \cup U_0 \cup \dots \cup U_{d-1}$ such that $\wsum{w'}(Y) = ip$.
Let  $G'$ be the graph with the set of vertices $V(G) \cup U_Z \cup U_0 \cup \dots \cup U_{d-1}$ and hyperedges inherited from $G$ plus these defined above.
We set $t' = t_p + n^d \cdot p$.

Suppose now that $X \subseteq V$ forms a~hyperclique of total weight $t$ in $G$.
Then $\wsum{w'}(X) = t_p + ip$ for some $i \in [0, n^d)$.
By the argument above, we can find a set $Y \subseteq U_Z \cup U_0 \cup \dots \cup U_{d-1}$ such that $\wsum{w'}(X\cup Y) = t'$ and $X \cup Y$ is \bmp{a hyperclique} in $G'$.

In the other direction, suppose we have successfully applied Lemma~\ref{lem:prime-hashing} and there is a~hyperclique $X' \subseteq V(G')$ with total weight  $t'$.
Then $p$ divides $\wsum{w'}(X' \setminus V)$, so \bmp{since all hyperedges intersecting both~$V$ and~$X' \setminus V$ have weight~$0$, we have }
$\wsum{w'}(X' \cap V) \equiv t \pmod{p}$ and
$\wsum{w}(X' \cap V) = t$, which gives a~desired hyperclique in $G$. 

The new instance has $\Oh(n)$ vertices and $\Oh(n^d)$ edges.
The weight range is $[0,n^d\cdot p]$ and, since $p = 2^{\Oh(n)}$, each weight can be encoded with $\Oh(n)$ bits.
The claim follows.
\end{proof}

We obtain \cref{thm:eewc:ub} as a corollary by taking $d=2$.

\subsection{Turing kernel for Max Weighted Hyperclique}

We turn our attention to the maximization variant of the weighted hyperclique problem.
We consider the problem \textsc{Max-Edge-Weight $d$-Uniform Hyperclique}, which takes the same input as \eewd, but the goal is to detect a hyperclique of total weight greater or equal to the target value $t$.
Even though we are not able to compress the weight function as in Theorem~\ref{thm:eewd-upper-bound}, we present a Turing kernelization with the same size.
We rely on a generic technique of reducing interval queries to exact queries. 

\begin{theorem}[\cite{NederlofLZ12}, Theorem 1]\label{thm:interval-reduction}
Let $U$ be a set of cardinality $n$, let 
$w \colon U \rightarrow \nn_0$ be a weight function, and let $l < u$ be non-negative integers
with $u - l > 1$.
There is a polynomial-time algorithm that returns a set
of pairs $\Omega = {(w_1, t_1), \dots ,(w_K, t_K)}$ with $w_i :
U \rightarrow \nn_0$ and
integers $t_1, t_2, \dots, t_K$, such that:
\begin{enumerate}
    \item $K$ is at most $(5n + 2)\cdot \log(u - l)$,
    \item for every set $X \subseteq U$ it holds that $w(X) \in [l, u]$ if and only if there exist $i \in [1, K]$ such that $w_i(X) = t_i$.
\end{enumerate}
\end{theorem}

A Turing kernel for a parameterized problem $\mathcal{P}$ is a polynomial-time algorithm that \bmp{decides any instance of} $\mathcal{P}$ with an access to an oracle that can solve instances (of possibly different \bmp{NP-problem}) of size polynomial with respect to the parameter.
The size of a Turing kernel is the maximal size of the instances \bmp{queried} to the oracle.
Observe that the number of calls to the oracle can be arbitrarily large.
The following theorem gives a one-sided error randomized Turing kernel of size $\Oh(n^{d+1})$ for \textsc{Max-Edge-Weight $d$-Uniform Hyperclique} parameterized by the number of vertices $n$.

\begin{theorem}\label{thm:eewd-ub-proof}
There is a randomized polynomial-time algorithm that, given an $n$-vertex instance $(G,w,t)$ of \textsc{Max-Edge-Weight $d$-Uniform Hyperclique},
returns a family of $K$ instances $(G_i,w_i,t_i)$, $i \in [K]$, of \eewd, each of bitsize $\Oh(n^{d+1})$, such that:
\begin{enumerate}
    \item $K$ is polynomial with respect to the input size,
    \item if $(G,w,t)$ is a \textsc{yes}-instance, then at least one instance $(G_i,w_i,t_i)$ is a \textsc{yes}-instance,
    \item if $(G,w,t)$ is a \textsc{no}-instance, then
    with probability $1 - 2^{-\Omega(n)}$ all the instances
    $(G_i,w_i,t_i)$ are \textsc{no}-instances.
\end{enumerate}
\end{theorem}
\begin{proof}
We apply Theorem~\ref{thm:interval-reduction} with $U$ being the set of hyperedges in $G$, $l = t$, and $u = n^d\cdot \max_{e\in E(G)} w_e$.
We can assume that $l \le u$, as otherwise the can be no solution.
We obtain $K = \log(u - l) \cdot \Oh(n^d)$ many weight functions $w^i$ and integers $t_i$, so that for each $X \subseteq U$ it holds
$\wsum{w}(X) \ge t$ if and only if
$\wsum{w^i}(X) = t_i$ for some $i \in [K]$.
Observe that $\log(u - l)$ is upper bounded by the input size, so the condition (1) is satisfied.

The original problem thus reduces to a disjunction of polynomially many instances of \eewd.
We use Theorem~\ref{thm:eewd-upper-bound} to compress each of them to $\Oh(n^{d+1})$ bits.
The probability that a single instance would be incorrectly compressed is bounded by $2^{-n}$.
By the union bound, the probability that any instance would be incorrectly compressed is $n^{\Oh(1)}\cdot 2^{-n} = 2^{-\Omega(n)}$.
\end{proof}

\subsection{Kernel upper bounds for CSP}

Similarly as for \eewdshort, we are able to prove a~tight randomized kernel upper bound for each considered \csp.
At first, we need to reduce the problem to the case where there are only $\Oh(n^{\deg(\Gamma)})$ weights, so we could encode each of them in $\Oh(n)$ bits.
To this end, we are again going to take advantage of the reduction framework from~\cite{JansenW20}.
Let $||\Phi||$ denote the sum of absolute values of weights in the formula $\Phi$.

\begin{lemma}[\cite{JansenW20}, Theorem 5.1]\label{lem:deg2}
For any constraint language $\Gamma$, there is a polynomial-time algorithm that, given a formula $\Phi_1 \in \csp(\Gamma)$ on $n_1$ variables and an integer $t_1$, returns a formula
$\Phi_2 \in \csp(\Gamma)$ on $n_2$ variables and an integer $t_2$, such that
\begin{enumerate}
    \item $n_2 = \Oh(n_1)$,
    \item the number of constraint applications in $\Phi_2$  is $\Oh(n^{\deg(\Gamma)})$, 
    \item $||\Phi_2|| = ||\Phi_1|| \cdot n^{\Oh(1)}$,
    \item $\exists_x \Phi_1(x) = t_1 \Longleftrightarrow \exists_y \Phi_2(y) = t_2$,
    \item $\exists_x \Phi_1(x) \ge t_1 \Longleftrightarrow \exists_y \Phi_2(y) \ge t_2$.
\end{enumerate}
\end{lemma}

Next, we proceed analogously to the proof of Theorem~\ref{thm:eewd-ub-proof}.

\begin{theorem}\label{thm:csp-upper-bound}
For any constraint language $\Gamma$, there is a randomized polynomial-time algorithm that, given a formula $\Phi_1 \in \csp(\Gamma)$ on $n$ variables and an integer $t_1$,
outputs a formula $\Phi_2 \in \csp(\Gamma)$ of bitsize $\Oh(n^{\deg(\Gamma)+1})$ and an integer $t_2$, such that:
\begin{enumerate}
    \item if $(\Phi_1, t_1)$ is a \textsc{yes}-instance of \ewcspg, then $(\Phi_2, t_2)$ is always a \textsc{yes}-instance of \ewcspg,
    \item if $(\Phi_1, t_1)$ is a \textsc{no}-instance, then $(\Phi_2, t_2)$ is a \textsc{no}-instance with probability at least $1 - 2^{-n}$.
\end{enumerate}
\end{theorem}
\begin{proof}
By Lemma~\ref{lem:deg2} we can assume that $\Phi$ has $\Oh(n^{\deg(\Gamma)})$ many constraint applications.
Let $V$ denote the set of variables in $\Phi_1$, $\mathcal{C}$ the set of constraint applications, and \bmp{let} $w \colon \mathcal{C} \to \zz$ be the weight function.
Let us
define $N = \max_{C \in \mathcal{C}} |w(C)|$. 
We can assume $\log N \le 2^n$, because otherwise the input length is lower bounded by $2^n$ and the brute-force algorithm becomes polynomial.

We apply Lemma~\ref{lem:prime-hashing} to the weight function $w$, target $t_1$, and $\eps = 2^{-n}$, to compute the desired prime~$p \le 2^n \cdot \text{poly}(n, \log N, \eps^{-1}) = 2^{\Oh(n)}$.
If there exists an assignment $x$ satisfying $\Phi_1(x) = t_1$, then $\Phi_1(x) \equiv t_1 \pmod{p}$.
Furthermore, with probability $1-2^{-n}$, the implication in the other direction holds as well.
In particular, in this case $p$ does not divide~$t_1$. 

Let $d = \deg(\Gamma)$ and $|\Gamma|$ be the number of relations in the language $\Gamma$.
Let us construct a new instance of \ewcspg with weights bounded by $p\cdot n^d\cdot |\Gamma|$.
For $C \in \mathcal{C}$,
we set $w'(C) = w(C) \pmod p$ and $t_p = t_1 \pmod p$.
Let $\Phi'_1$ be obtained from $\Phi$ by replacing the weight function by $w'$.
The condition $\Phi_1(x) \equiv t_1 \pmod{p}$ is equivalent to \bmp{the existence of $i \in [0, n^d\cdot |\Gamma|)$ for which} $\Phi'_1(x) = t_p + ip$, because for each $d$-tuple of variables there are at most $|\Gamma|$ constraint applications and thus the sum comprises of at most $n^d\cdot |\Gamma|$ summands from the range $[0,p)$.

We can assume that $\Gamma$ contains some satisfiable constraint $f$
as otherwise the only feasible target value is $t=0$. 
We introduce a set $U_Z$ of $\ar(f)-1$ new variables and for each $j \in [0,d-1]$ we introduce a set $U_j$ of $n\cdot |\Gamma|$ new variables.
For every $v \in U_j$ we create a constraint application $\langle f, (U_Z, v), n^j\cdot p \rangle$.
Observe that for every integer $i \in [0, n^d\cdot |\Gamma|]$, we can find an assignment to the variables $U_Z \cup U_0 \cup \dots \cup U_{d-1}$ with total weight $i\cdot p$.
Let $V_2 = V \cup U_Z \cup U_0 \cup \dots \cup U_{d-1}$ and $\Phi_2$ be the formula with the set of variables $V_2$ and constraints applications $\mathcal{C}$ plus these defined above.
We set $t_2 = t_p + p \cdot n^d \cdot |\Gamma|$.

Consider an assignment $x$ to $V$ such that $\Phi_1(x) = t_1$.
It holds that $\Phi'_1(x) = t_p + ip$ for some $i \in [0, n^d\cdot |\Gamma|)$.
By the argument above, we can find an assignment $y$ to $V_2$, which coincides with $x$ on $V$, so
that $\Phi_2(y) = t_2$.

In the other direction, suppose we have successfully applied Lemma~\ref{lem:prime-hashing} and there is an~assignment $x$ to $V_2$ such that $\Phi_2(x) = t_2$.
Since $p$ divides all the weights for constraint applications outside $\mathcal{C}$, the assignment $x'$ given by $x$ applied only to $V$ satisfies $\Phi'_1(x) \equiv t_1 \pmod{p}$ and so
$\Phi_1(x) = t_1$.

The new instance has $\Oh(n)$ variables and $\Oh(n^d)$ constraint applications.
The weight range is $[0,n^d\cdot p]$ and, since $p = 2^{\Oh(n)}$, each weight can be encoded with $\Oh(n)$ bits.
The claim follows.
\end{proof}

\section{Node-weighted Vertex Cover in bipartite graphs}
\label{sec:vc}
\subsection{Preserving all minimum solutions}

In this section we show that given a bipartite graph $G$ with $n$ vertices and weight function $w \colon V(G) \to \mathbb{N}_+$, we can compute (in polynomial time) a new weight function $w^*$ which assigns a positive integer weight of at most $n$ to all vertices of $G$ such that for all $S \subseteq V(G)$, $S$ is a minimum $w$-weighted vertex cover of $G$ if and only if $S$ is a minimum $w^*$-weighted vertex cover of $G$.

Next we define a well-studied combinatorial optimization problem for node-weighted graphs:  $b$-matching (cf.~\cite[Chapter 21]{Schrijver03}). We discuss some of its properties which will be useful later while developing reduction rules. For a vertex~$v$ in a graph~$G$, we denote by~$E_v$ the set of edges incident on~$v$. 

\begin{definition}[$b$-matching]
Let $(G, b)$ be a node-weighted graph with $b \colon V(G) \to  \mathbb{N}_+$. A $b$-matching is a function $z \colon E(G) \to \mathbb{N}_0$ such that $z(E_v) \leq b(v)$ for each vertex~$v \in V(G)$.
\end{definition}

A~$b$-matching~$z \colon E(G) \to \mathbb{N}_0$ is said to be \emph{maximum} if there does not exist a $b$-matching~$z^*$ satisfying~$z^*(E(G))>z(E(G))$.

\begin{theorem}
\label{theorem:b-matching_computation}
Let~$(G, b)$ be a node-weighted bipartite graph with~$b \colon V(G) \to  \mathbb{N}_+$. Then a maximum $b$-matching can be computed in $\mathcal{O}(m n)$ time, where $m = |E(G)|$ and $n = |V(G)|$.
\end{theorem}

\begin{proof}
For bipartite graphs, finding a maximum $b$-matching can be reduced (in linear time) to the problem of finding maximum flow, where the size of the flow network remains asymptotically the same {\cite[Section 21.13a, page 358]{Schrijver03}}. A maximum flow can be found in  $\mathcal{O}(m n)$ time due to Orlin \cite{Orlin13}. Hence a maximum $b$-matching can be computed in $\mathcal{O}(m n)$ time. 
\end{proof}

\bmp{In the literature, the function giving vertex capacities for a matching is typically called~$b$. As we will apply this concept when using the vertex weight function~$w$ to prescribe capacities, from now on we will refer to such matchings as~$w$-matchings for a given node-weight function~$w\colon V(G) \to \mathbb{N}_+$.}

The following is a weighted version of K\H{o}nig's theorem.

\begin{theorem}[{\cite[Corollary 21.1a]{Schrijver03}}]
\label{theorem:Weighted_version_of_Konig's_theorem}
Let~$(G,w)$ be a node-weighted bipartite graph. The minimum weight of a vertex cover of~$G$ is equal to the value of a maximum $w$-matching in~$G$.
\end{theorem}

\bmp{We can use Theorem~\ref{theorem:Weighted_version_of_Konig's_theorem} to infer some properties of minimum-weight vertex covers.}

\allowdisplaybreaks

\begin{lemma}
\label{lemma:weight_reduces_by_1}
If $z \colon E(G) \to \mathbb{N}_0$ is a maximum $w$-matching of~$(G,w)$, with~$z(\{u,v\}) > 0$ for some edge~$\{u,v\} \in E(G)$, then any minimum-weight vertex cover of~$(G,w)$ contains \emph{exactly one} of~$\{u,v\}$.
\end{lemma}

\begin{claimproof}
Note that to cover edge~$\{u, v\}$ at least one of~$\{u, v\}$ must be in any vertex cover. Assume for a contradiction that there exists a set~$C \subseteq V(G)$ such that~$C$ is a minimum-weight vertex cover of~$(G, w)$ containing both~$u$ and~$v$. Observe that the sum~$\sum _{c \in C} z(E_c)$ contains each edge value~$z(\{x,y\})$ for~$\{x,y\} \in E(G)$ at least once since~$C$ is a vertex cover, and this sum contains the value~$z(\{u,v\})$ twice since the edge~$\{u,v\}$ is contained in \emph{both}~$E_v$ and~$E_u$ with~$u,v \in C$. Thus we have:
\begin{align*} 
    w(C) &= \sum _{v \in C} w(v) & \mbox{By definition of~$w(C)$} \\
    &\geq \sum _{v \in C} z(E_v) & \mbox{By definition of $w$-matching} \\
    &\geq 2 z(\{u,v\}) + \sum _{e \in E(G) \setminus \{\{u,v\}\}} z(e) & \mbox{$C$ is a vertex cover and~$u,v \in C$} \\
    &\geq z(\{u,v\}) + \sum _{e \in E(G)} z(e) & \mbox{Collecting terms} \\
    &> \sum _{e \in E(G)} z(e) = z(E(G)). & \mbox{Since~$z(\{u,v\}) > 0$}
\end{align*}
\bmp{This contradicts~Theorem~\ref{theorem:Weighted_version_of_Konig's_theorem}.}
\end{claimproof}

\bmp{The next lemma gives a condition under which no optimal solution contains a given vertex~$v$. Recall that~$E_v$ is the set of edges incident on~$v$.}

\begin{lemma}
\label{lemma:weight_stays_the_same}
If $z$ is a maximum $w$-matching of~$(G,w)$ with~$w(v) > z(E_v)$ for a vertex~$v$, then no minimum-weight vertex cover of~$(G,w)$ contains~$v$. \end{lemma}

\begin{claimproof}
Assume for a contradiction that there exists a minimum-weight vertex cover $C \subseteq V(G)$ such that $v \in C$. Then we have:
\begin{align*}
    w(C) &= \sum_{c \in C} w(c) &\mbox{By definition of w(C)}\\
    &= w(v) + \sum_{c \in C \setminus\{v\}} w(c) &\mbox{Distributing terms}\\
    &\geq w(v) + \sum_{c \in C \setminus \{v\}} z(E_c) &\mbox{By definition of $w$-matching}\\
    &> z(E_v) + \sum_{c \in C \setminus \{v\}} z(E_c) &\mbox{Since $w(v) > z(E_v)$}\\
    &> \sum_{e \in E(G)} z(e) = z(E(G)).
\end{align*}
This contradicts Theorem~\ref{theorem:Weighted_version_of_Konig's_theorem}.
\end{claimproof}

Now we present reduction rules which can reduce the weights of a node-weighted bipartite graph~$(G, w)$ without changing the set of optimal vertex covers.

\begin{reductionrule}
\label{rr:twoedges} 
 If $z \colon E(G) \to \mathbb{N}_0$ is a maximum $w$-matching of~$(G,w)$ with~$z(\{x,y\}) > 1$ for some edge~$\{x,y\} \in E(G)$, then obtain a new weight function $w^*\colon V(G) \to \mathbb{N}_+$ from $w$ as follows.
\begin{equation}
\label{equation:rr_1:reduced_node_weight_function}
    w^*(v) = \begin{cases}
    w(v) - (z(x, y) - 1)  & \mbox{if $v \in \{x, y\}$}\\
    w(v) & \mbox{otherwise}
    \end{cases}
\end{equation}
\end{reductionrule}

\bmp{We now prove that this reduction rule is correct. Recall that~$\mathcal{C}(G, w)$ is the set containing all minimum~$w$-weighted vertex covers of graph $G$.}

\begin{theorem}
\label{theorem:MWVC_1_remains_the_same}
If $(G,w)$ is reduced to $(G,w^*)$ by Reduction Rule~\ref{rr:twoedges}, then $\mathcal{C}(G,w) = \mathcal{C}(G,w^*)$.
\end{theorem}

\begin{proof}
First we show how~$z$ can be transformed into an optimal~$w^*$-matching of~$(G,w^*)$.

\begin{claim}
\label{claim:rr_1:new_small_weight_reduction}
Let~$z$ be a maximum $w$-matching of~$(G,w)$ and~$z(\{x,y\}) > 1$ for edge~$\{x,y\} \in E(G)$, then the function $z^* \colon E(G) \to \mathbb{N}_0$ defined as:
    \begin{equation}
    \label{equation:new_matching_1}
        z^*(e) = \begin{cases}
             z(e) - (z(e) - 1) = 1 & \mbox{if~$e = \{x,y\}$}\\
            z(e) & \mbox{otherwise}
        \end{cases}
    \end{equation}
    is a maximum $w^*$-matching of~$(G, w^*)$, where $w^*$ is defined by~Equation~$(\ref{equation:rr_1:reduced_node_weight_function})$.
\end{claim}

\begin{claimproof}
Assume for a contradiction that $z^*$ is not a maximum $w^*$-matching of $(G, w^*)$ and let $z'\colon E(G) \to \mathbb{N}_0$ be a function such that $z'$ is a maximum $w^*$-matching. Consequently, $z'(E(G)) > z^*(E(G))$.
We construct a new weight function $\hat{z} \colon E(G) \to \mathbb{N}_0$ from $z'$ as follows.
\begin{equation}
\label{equation:new_matching_2}
    \hat{z}(e) = \begin{cases}
    z'(e) + (z(e) -1)  & \mbox{if $e = \{x, y\}$}\\
    z'(e) & \mbox{otherwise}
    \end{cases}
\end{equation}
Observe that $\hat{z}$ is a $w$-matching of $(G, w)$. Moreover, $z^*(E(G)) = z(E(G)) - (z(e) - 1)$ and $\hat{z}(E(G)) = z'(E(G)) + (z(e) - 1)$ due to Equation~$(\ref{equation:new_matching_1})$ and Equation~$(\ref{equation:new_matching_2})$, respectively. But by assumption~$z'(E(G)) > z^*(E(G))$, and thus we have $\hat{z}(E(G)) > z(E(G))$, a contradiction to the fact that $z$ is a maximum~$w$-matching of $(G, w)$.
\end{claimproof}

Next we prove the following two claims which we will use later in order to prove the main theorem.
\begin{claim}
\label{claim:MWVC_1}
For each $C \in \mathcal{C}(G, w)$ we have $w^*(C) = w(C) - (z(x, y) - 1)$.
\end{claim}

\begin{claimproof}
Since $z(\{x, y\}) > 1$, for each $C \in \mathcal{C}(G, w)$ we have $|C \cap \{x, y\}| = 1$ due to Lemma~\ref{lemma:weight_reduces_by_1}. Thus by Equation~$(\ref{equation:rr_1:reduced_node_weight_function})$, we have $w^*(C) = w(C) - (z(x, y) - 1)$.
\end{claimproof}

\begin{claim}
\label{claim:MWVC_2}
For each $C \in \mathcal{C}(G, w^*)$ we have $w(C) = w^*(C) + (z(x, y) -1)$.
\end{claim}
\begin{claimproof}
The function~$z^* \colon E(G) \to \mathbb{N}$ is a maximum~$w^*$-matching of~$(G, w^*)$ with~$z^*(\{x, y\}) > 0$ due to Claim~\ref{claim:rr_1:new_small_weight_reduction}. Therefore by~Lemma~\ref{lemma:weight_reduces_by_1}  we have~$|C \cap \{x, y\}| = 1$, implying~$w(C) = w^*(C) + (z(x, y) - 1)$ due to  Equation~$(\ref{equation:rr_1:reduced_node_weight_function})$.
\end{claimproof}

Using these claims, we prove Theorem~\ref{theorem:MWVC_1_remains_the_same}. Note that $C$ is a vertex cover of $(G, w)$ if and only if $C$ is a vertex cover of $(G, w^*)$ because the underlying bipartite graph $G$ remains the same before and after applying Reduction Rule~\ref{rr:twoedges}; only the weights change during the reduction rule.

To show $\mathcal{C}(G, w) \subseteq \mathcal{C}(G, w^*)$, choose any $C \in \mathcal{C}(G, w)$. Assume for a contradiction that $C \notin \mathcal{C}(G, w^*)$, and pick an arbitrary $C' \in \mathcal{C}(G, w^*)$. Consequently, $w^*(C') < w^*(C)$. Also note that $w^*(C) = w(C) - (z(x, y) - 1)$, and $w^*(C') = w(C') - (z(x, y) - 1)$ due to Claim~\ref{claim:MWVC_1} and  Claim~\ref{claim:MWVC_2}, respectively. These  collectively imply that $w(C') < w(C)$, a contradiction to the fact that $C \in \mathcal{C}(G, w)$.

To show that $ \mathcal{C}(G, w^*) \subseteq  \mathcal{C}(G, w)$, choose any $C \in  \mathcal{C}(G, w^*)$. Assume for a contradiction that $C \notin  \mathcal{C}(G, w)$, and pick an arbitrary $C' \in  \mathcal{C}(G, w)$. Consequently, $w(C') < w(C)$. Moreover, $w(C') = w^*(C') + (z(x, y) - 1)$ and $w (C) = w^*(C) + (z(x, y) - 1)$ by Claim~\ref{claim:MWVC_1} and Claim~\ref{claim:MWVC_2}, respectively. These  collectively imply that $w^*(C') < w^*(C)$, a contradiction to the fact that $C \in \mathcal{C}(G, w^*)$.

This concludes the proof of Theorem~\ref{theorem:MWVC_1_remains_the_same}.
\end{proof}

Next we \mic{present} the following lemma which gives a characterization of vertices which have large weights after exhaustive application of  Reduction Rule~\ref{rr:twoedges}. 
\begin{lemma}
\label{lemma:high_node_weighted_vertex_is_not_saturated}
Let~$(G, w)$ be a node-weighted bipartite graph with a maximum~$w$-matching $z \colon E(G) \to \mathbb{N}$. If Reduction Rule~\ref{rr:twoedges} is not applicable on $(G, w)$ with respect to $z$, and there is a vertex $v \in V(G)$ with $w(v) > |V(G)| - 1$, then $w(v) > z(E_v)$.
\end{lemma}

\begin{proof}
Since Reduction Rule~\ref{rr:twoedges} is not applicable with respect to~$z$, we have~$z(\{u, v\}) \leq 1$ for each vertex $u \in N_G(v)$. As~$|N_G(v)| \leq  |V(G)| - 1$, we have~$z(E_v) \leq |V(G)| - 1$. By assumption we have~$w(v) > |V(G)| - 1$, and hence $w(v) > z(E_v)$. 
\end{proof}

Based on the above lemma we give the following reduction rule which will further reduce the weight of vertices having weight larger than $|V(G)|$.

\begin{reductionrule}
\label{rr:D}
If $z \colon E(G) \to \mathbb{N}$ is a maximum $w$-matching of~$(G,w)$ with~$w(v) > |V(G)|$ for a vertex~$v$, and Reduction Rule~\ref{rr:twoedges} is not applicable with respect to~$z$, then obtain a new weight function $w^*\colon V(G) \to \mathbb{N}_+$ from $w$ as follows.
\begin{equation}
\label{equation:rr_2:reduced_node_weight_function}
    w^*(x) = \begin{cases}
    |V(G)| & \mbox{if $x = v$}\\
    w(x) & \mbox{otherwise.}
    \end{cases}
    \end{equation}
\end{reductionrule}

We continue by proving the correctness of Reduction Rule~\ref{rr:D}.

\begin{theorem}
\label{theorem:MWVC_2_remains_the_same}
If $(G,w)$ is reduced to $(G,w^*)$ by Reduction Rule~\ref{rr:D}, then $\mathcal{C}(G,w) = \mathcal{C}(G,w^*)$.
\end{theorem}

\begin{proof}
First we see the following claim which we will use later in order to prove the theorem.
\begin{claim}
\label{claim:rr_2:new_small_node_weight_reduction}
If~$z$ is a maximum $w$-matching of~$(G,w)$ and~$w(v) > z(E_v)$, then~$z$ is also a maximum $w^*$-matching.
\end{claim}

\begin{claimproof}
Assume for a contradiction that $z$ is not a maximum $w^*$-matching, and let there be a function $\Tilde{z} \colon E(G) \to \mathbb{N}$ such that $\Tilde{z}$ is a maximum $w^*$-matching of $(G, w)$. Consequently, $\Tilde{z}(E(G)) > z(E(G))$. Observe that every $w^*$-matching is also a $w$-matching because any matching which satisfies all the constraints under a smaller weight function also satisfies all the constraints under a larger weight function.
Thus $\Tilde{z}$ is also a maximum $z$-matching of $(G, w)$ with $\Tilde{z}(E(G)) > z(E(G))$, a contradiction to the fact that $z$ is a maximum $w$-matching of $(G, w)$.   
\end{claimproof}

Next we prove the following two claims which will be used later in order to prove the theorem.
\begin{claim}
\label{claim:MWVC_3}
For each $C \in \mathcal{C}(G, w)$ we have $w^*(C) = w(C)$.
\end{claim}

\begin{claimproof}
Directly from Lemma~\ref{lemma:weight_stays_the_same} and Equation~$(\ref{equation:rr_2:reduced_node_weight_function})$.
\end{claimproof}

\begin{claim}
\label{claim:MWVC_4}
For each $C \in \mathcal{C}(G, w^*)$ we have $w(C) = w^*(C)$.
\end{claim}

\begin{claimproof}
Observe that in $(G, w^*)$ we have $w^*(v) = |V(G)| > |V(G)| - 1$, implying $w^*(v) > z(E_v)$ due to Lemma~\ref{lemma:high_node_weighted_vertex_is_not_saturated}. Which implies that $v \notin C$ due to Lemma~\ref{lemma:weight_stays_the_same}, and hence by Equation~$(\ref{equation:rr_2:reduced_node_weight_function})$, we have $w(C) = w^*(C)$.
\end{claimproof}

Now we proceed to prove Theorem~\ref{theorem:MWVC_2_remains_the_same}. First note that $C$ is a vertex cover of $(G, w)$ if and only if $C$ is a vertex cover of $(G, w^*)$ because the underlying graph $G$ does not change when the Reduction Rule~\ref{rr:D} is applied, only the weight of vertex $x$ changes. 

To show $\mathcal{C}(G, w)\subseteq \mathcal{C} (G, w^*)$, let $C \in \mathcal{C}(G, w)$. Assume for a contradiction that $C \notin \mathcal{C}(G, w^*)$, and consider an arbitrary $C' \in \mathcal{C}(G, w^*)$. Consequently, $w^*(C') < w^*(C)$. Furthermore, we have $w^*(C) = w(C)$ and $w(C') = w^*(C')$ by Claim~\ref{claim:MWVC_3} and \ref{claim:MWVC_4}, respectively. Which collectively imply that $w(C') < w(C)$, a contradiction to that fact that $C$ is a minimum-weight vertex cover of $(G, w)$.

To show $\mathcal{C}(G, w^*)\subseteq \mathcal{C} (G, w)$, let $C \in \mathcal{C}(G, w^*)$. Assume for a contradiction that $C \notin \mathcal{C}(G, w)$, and consider an arbitrary $C' \in \mathcal{C}(G, w)$. Consequently, $w(C') < w(C)$, furthermore, we have $w^*(C') = w(C')$ and $w(C) = w^*(C)$ due to Claim~\ref{claim:MWVC_3} and \ref{claim:MWVC_4}, respectively. Which collectively imply that~$w^*(C') < w^*(C)$, a contradiction to that fact that $C$ is a minimum-weight vertex cover of $(G, w^*)$.

This concludes the proof of Theorem~\ref{theorem:MWVC_2_remains_the_same}.
\end{proof}

Combining Reduction Rule~\ref{rr:twoedges} and Reduction Rule~\ref{rr:D}, we have the main result of this section.

\bipartiteWeightStatement*

\begin{proof}
\bmp{The algorithm starts by computing a maximum $w$-matching~$z$ in strongly polynomial time, using Theorem~\ref{theorem:b-matching_computation}. 

As long as Reduction Rule~\ref{rr:twoedges} can be applied with respect to~$z$ to some edge~$\{x,y\}$, we update the weight function~$w$ as indicated by Equation~\eqref{equation:rr_1:reduced_node_weight_function} and update~$z$ as indicated by Claim~\ref{claim:rr_1:new_small_weight_reduction}, so that~$z$ again becomes a maximum $w$-matching for the updated weight function~$w$. As an application of the rule to edge~$\{u,v\}$ reduces the value of the edge to~$1$, the rule can be applied at most~$|E(G)|$ times. By Theorem~\ref{theorem:MWVC_1_remains_the_same}, the set of minimum-weight vertex covers is preserved by this process. This phase terminates at the point that we can no longer apply Reduction Rule~\ref{rr:twoedges} for the current maximum $w$-matching~$z$.

Next, we exhaustively apply Reduction Rule~\ref{rr:D} with respect to~$z$. As a consequence, we bound the weight of each vertex by at most $|V(G)| = n$. In the worst case, Reduction Rule~\ref{rr:D} is applied at most $n$ times, once for each vertex. By Theorem~\ref{theorem:MWVC_2_remains_the_same}, the set of minimum-weight vertex covers is preserved by this phase. The resulting weighted graph is given as the output of the algorithm.}
\end{proof}

The following lemma shows that the weight-reduction to a value in the range~$\{1, \ldots, n\}$ is best-possible.

\begin{lemma}
\label{lemma:weight_reduction_is_tight}
For each $n \geq 1$, there exists a node-weighted bipartite graph~$G_n$ on~$n$ vertices with weight function~$w \colon V(G_n) \to \mathbb{N}_+$ such that for all weight functions $w' \colon V(G_n) \to \mathbb{N}_+$ with~$\mathcal{C}(G_n, w) = \mathcal{C}(G_n, w')$ we have:
\[
\max\limits_{v \in V(G_n)} (w'(v)) \geq |V(G_n)| = n.
\]
\end{lemma}

\begin{claimproof}
For given $n \geq 1$, let $G_n$ be the star graph on $n$ vertices with~$n-1$ leaves. The weight function $w\colon V(G_n) \to \mathbb{N}_+$ is defined as follows.
\begin{equation}
    w(v) = \begin{cases}
    n  & \mbox{if $v$ is the center vertex}\\
    1 & \mbox{otherwise}
    \end{cases}
\end{equation}
As a vertex cover for~$G_n$ has to contain the center vertex (total weight~$n$) or all the~$n-1$ leaves (total weight~$n-1$), the collection $\mathcal{C}(G_n, w)$ contains a single set containing all $n-1$ leaf nodes of $(G_n , w)$. The unique minimum-weight vertex cover of~$G_n$ has weight~$n-1$.

Now, assume for a contradiction that there is a weight function $w'\colon V(G_n) \to [n-1]$ with $\mathcal{C}(G_n, w) = \mathcal{C}(G_n, w')$. Since $w'$ assigns an integer weight between 1 and $n-1$ to each vertex of $G_n$, the weight of center node is at most $n-1$, whereas the weight assigned to each leaf node must be at least 1. Thus, the singleton set containing the center node is a minimum-weight vertex cover of $(G_n, w')$, whereas it is not for $(G_n, w)$. This contradicts the fact that $\mathcal{C}(G_n, w) = \mathcal{C}(G_n, w')$.
\end{claimproof}

Lemma~\ref{lemma:weight_reduction_is_tight} implies that the bound achieved by Theorem~\ref{theorem:bipartiteweightreduction} is essentially tight for infinitely many node-weighted bipartite graphs.\par

\subsection{Preserving the relative weight of solutions}

Suppose that we want a stronger guarantee for weight compression, so that we
 keep the information for each pair of inclusion-minimal vertex covers which one of them is lighter.
It turns out that with this requirement,
one cannot compress the weights to such a small range as before.
We provide a lower bound on the number of non-equivalent weight functions, which implies
that cannot encode weights with a small number of bits.
We begin with introducing the notion of a~threshold function.

\begin{definition}[Threshold function]
An~$n$-bit threshold function is a Boolean function $f \colon \{0, 1\}^n \to \{0, 1\}$ \mic{for which there exist a} weight vector~$w_f \in {\mathbb{Z}}^n$ and a threshold~$t_f \in \mathbb{Z}$ such that for any given input~$X \in \{0, 1\}^n$, we have:
\begin{equation}
    f(X) = \begin{cases}
    1  & \mbox{if $\sum_{i=1}^{n} w_f[i] X[i] \geq t_f$}\\
    0 & \mbox{otherwise.}
    \end{cases}
\end{equation}
\mic{We say that such a threshold function $f$ is induced by the pair $(w_f, t_f)$.}
Two threshold functions~$f$ and~$g$ are \emph{distinct} if there exists~$X \in \{0, 1\}^n$ such that~$f(X) \neq g(X)$.
\end{definition}
Next, we look at the following theorem which gives a bound on the number of threshold functions. 
\begin{theorem}[\cite{Irmatov20}]
\label{theorem:BoundOnThresholdFunctions}
There exists a constant $\alpha > 0$ such that for each~$n \geq 1$ the number of distinct~$n$-bit threshold functions is at least~$2^{\alpha n^2}$.
\end{theorem}
Before we proceed further we remark that the above bound is originally given for threshold functions with \emph{real valued} weights and threshold, with~$-1/+1$ coefficients, but without loss of generality the same bound also holds if we restrict weights and threshold to only take \emph{integral values} and use~$0/1$-coefficients~\cite[page 67]{BabaiHPS10}.
\bmp{
\begin{definition}[\sk{Vertex-cover equivalent functions}]
Let~$G$ be an undirected graph and let~$w, w' \colon V(G) \to \mathbb{N}_+$ be weight functions. The weight functions~$w$ and~$w'$ are \emph{vertex-cover equivalent} if for all pairs of inclusion-minimal vertex covers~$S_1, S_2 \subseteq V(G)$ of~$G$, the following equivalence holds:
\begin{equation*}
w(S_1) \leq w(S_2) \Leftrightarrow w'(S_1) \leq w'(S_2).    
\end{equation*}
\end{definition}
}
Now we are all set to present the theorem which gives a lower bound on how much can we compress the weights of node-weighted bipartite graphs, if we want to preserve the order among \mic{all the} vertex covers. 
\vcThresholdCountState*
\begin{proof}
\sk{Assume for a contradiction that for each~$n\geq 1$ and for each node-weighted bipartite graph~$(G_n, w)$ with weight function~$w: V(G_n) \to \mathbb{N}_+$, there exists a small weight function $w'$
which assign weights from the range~$\{1, 2 , \cdots, 2^{o(n)}\}$ while maintaining the property that~$w$ and~$w'$ are vertex-equivalent.}

At a high level, we give an injective (one-to-one) function from the set of threshold functions to the set of equivalence classes (with respect to vertex-cover equivalence) of  weight functions on bipartite graphs.
We further show that the number of different node-weighted bipartite graphs with the small weight function is strictly smaller compared to the number of threshold functions, which gives a contradiction. 

More formally, \mic{let $G_n$ be a $2(n+1)$-vertex bipartite graph with the set of vertices $v_1, \dots v_{n+1}, v'_1, \dots, v'_{n+1}$ and $n+1$ edges of the form $v_iv'_i$.
Given a weight vector~$w \in {\mathbb{Z}}^n$ and threshold~$t \in \mathbb{Z}$, we construct a weight function $h_{w,t} \colon V(G_n) \to \mathbb{N}_+$ as follows.}
\begin{enumerate}
    \item For each~$i \in [n]$, the weight of vertex~$v_i$ is~$w[i] + c$, where~$c = \left\lvert {\min\limits_{i \in [n]}(w[i], t)} \right\lvert + 1 $. 
    
    \item The weight of vertex~$v_{n+1}$ is set to~$t + c$.

    \item For each~$i \in [n+1]$, the weight of vertex~$v'_i$ is set to~$c$.
\end{enumerate}
We have added the extra term~$c$ to the weight of each vertex to ensure that each vertex has positive weight.

In the following claim we show how an input~$X \in \{0, 1\}^n$ to~$n$-bit threshold function is used to read out a vertex cover of~$G_n$.
\begin{claim}
\label{claim:thresholdtovcs}
For a given~$X \in \{0, 1\}^n$ 
there exist inclusion-minimal vertex covers~$S_1$ and~$S_2$ of~$G_n$ such that for any
weight vector $w \in \zz^n$ and threshold $t \in \zz$
we have~$h_{w,t}(S_1) =  \sum_{i=1}^{n} w[i]  X[i] + (n + 1) c$ and~$h_{w,t}(S_2) = t + (n + 1)  c$, where~$c = \left\lvert {\min\limits_{i \in [n]}(w[i], t)} \right\lvert + 1 $.
\end{claim}
\begin{claimproof}
Given an input~$X \in \{0, 1\}^n$, we construct~$S_1$ and~$S_2$  as follows.\\
\textbf{Construction of the vertex cover~$S_1$:}
\begin{enumerate}[i.]
    \item For each~$i \in [n]$, if~$X_i = 1$ then add vertex~$v_i$ to~$S_1$ to cover the edge~$e_i$ in the graph~$G_n$.
    
    \item For each~$i \in [n]$, if~$X_i = 0$ then add vertex~$v'_i$ to~$S_1$ to cover the edge~$e_i$ in the graph~$G_n$.
    
    \item Add vertex~$v'_{n+1}$ to~$S_1$ to cover the edge~$e_{n+1}$. 
\end{enumerate}
\textbf{Construction of the vertex cover~$S_2$:}
\begin{enumerate}[i.]
    \item For each~$i= [n]$, add vertex~$v'_i$ to $S_2$ to cover the edge~$e_i$.
    \item Add vertex~$v_{n+1}$ to $S_2$ to cover the edge~$e_{n+1}$.
\end{enumerate}
Note that in both~$S_1$ and~$S_2$, we select exactly one endpoint of each edge (and this selection only depends on the value of~$X$), therefore both~$S_1$ and~$S_2$ are minimal vertex covers of~$G_n$. Moreover, for each threshold function~$f$, it is easy to observe by the construction of~$(G_n, w)$ that~$w(S_1) = \sum_{i=1}^n w[i]  X[i] + (n + 1)  c$ and~$w(S_2) = t + (n + 1)  c$. Recall that the extra~$(n+1)c$ term in the addition comes from the fact that we have increased the node-weight of each vertex by~$c$ during the construction of~$(G_n, w)$ to ensure that each vertex has a positive node-weight.
\end{claimproof}

\mic{It remains to show that the given construction maps distinct threshold functions into vertex-cover non-equivalent weight functions, which provides a lower bound on the number of equivalence classes.}

\begin{claim}\label{lem:injection}
Suppose that $(w_f, t_f)$ and $(w_g, t_g)$ induce distinct threshold functions $f, g$.
Then the functions~$h_{w_f, t_f}$ and~$h_{w_g, t_g}$
are not vertex-cover equivalent with respect to $G_n$.
\end{claim}
\begin{claimproof}
Since~$f$ and~$g$ are distinct threshold functions, by definition there exists an input~$X \in \{0, 1\}^n$ such that~$f(X) \neq g(X)$, thereby one of the \bmp{following} must hold.
\begin{enumerate}[(a)]
    \item $\sum_{i=1}^n w_f[i]  X[i] \geq t_f$ and $\sum_{i=1}^n w_g[i]  X[i] < t_g$, or 
    \label{condition:one}
    
    \item $\sum_{i=1}^n w_f[i]  X[i] < t_f$ and $\sum_{i=1}^n w_g[i]  X[i] \geq t_g$.
    \label{condition:two}
\end{enumerate}
Now given such~$X$, we construct the vertex covers~$S_1, S_2 \subseteq V(G_n)$ of~$G_n$ using Claim~\ref{claim:thresholdtovcs}, such that~$h_{w_f, t_f}(S_1) = \sum_{i=1}^n w_f[i]  X[i] + (n + 1)  c_f$, $h_{w_f, t_f}(S_2) = t_f + (n + 1)  c_f$, $h_{w_g, t_g}(S_1) = \sum_{i=1}^n w_g[i]  X[i] + (n + 1)  c_g$, and~$h_{w_g, t_g}(S_2) = t_g  + (n + 1) c_g$.\\
Suppose w.l.o.g. that condition~$(\ref{condition:one})$ holds. Thus we have~$h_{w_f, t_f}(S_1) \geq h_{w_f, t_f}(S_2)$ whereas $h_{w_g, t_g}(S_1) < h_{w_g, t_g}(S_2)$, implying~$h_{w_f, t_f}$ and~$h_{w_g, t_g}$ are \sk{not vertex-cover equivalent}.
\end{claimproof}
\mic{
Let $\alpha$ be the constant from Theorem~\ref{theorem:BoundOnThresholdFunctions}.
Suppose that for each bipartite graph $G_n$ and weight function $h \colon V(G_n) \to \mathbb{N}_+$ there exists a vertex-cover equivalent function $h'\colon V(G_n) \to [2^{\frac{\alpha n}{3}}]$.
Then the number of distinct~vertex-cover equivalent functions on $G_n$ is~$2^{\frac{2\alpha n(n+1)}{3}}$, since there are~$2(n+1)$ vertices in~$G_n$ and each vertex can have~$2^{\frac{\alpha n}{3}}$ different weights.
By Claim~\ref{lem:injection} there is an injective mapping from $n$-bit threshold functions
to the vertex-cover equivalence classes on $G_n$ and
from Theorem~\ref{theorem:BoundOnThresholdFunctions} we know that the size of the first set is at least $2^{{\alpha n^2}}$. This leads to a contradiction with the assumption and
finishes the proof of Theorem~\ref{theorem:vclowerbound2}.}
\end{proof}
\end{document}